\documentclass{article}
\usepackage{fullpage}

\usepackage{times}
\usepackage{tikz}
\usetikzlibrary{decorations.pathreplacing}

\usepackage{amsthm}
\usepackage{amsmath}
\usepackage{amsfonts}
\usepackage{amssymb}

\newtheorem{theorem}{Theorem}
\newtheorem{lemma}{Lemma}

\newtheorem{proposition}{Proposition}
\theoremstyle{definition}
\newtheorem{definition}{Definition}
\newtheorem{remark}{Remark}
\newtheorem*{question}{Central Question}
\newtheorem{example}{Example}[]

\newtheorem{corollary}{Corollary}
\def\firstcircle{(150:1.75cm) circle (2.5cm)}
\def\secondcircle{(30:1.75cm) circle (2.5cm)}
\def\thirdcircle{(270:1.75cm) circle (2.5cm)}
\def\leftcircle{(180:1.5cm) circle (2.5cm)}
\def\rightcircle{(0:1.5cm) circle (2.5cm)}

\newcommand{\poly}{\mathrm{poly}}

\newcommand{\prob}{\mathrm{Prob}}
\newcommand{\hdist}{\mathrm{HammingDist}}
\newcommand{\eqp}{=^{+}}
\newcommand{\lep}{\mathbin{\le^{+}}}
\newcommand{\gep}{\mathbin{\ge^{+}}}

\begin{document}

\title{Communication Complexity of the Secret Key Agreement  in Algorithmic Information Theory}

\author{Emirhan G\"urp\i nar and Andrei Romashchenko}

\maketitle

\begin{abstract}
It is known that the mutual information, in the sense of Kolmogorov complexity, of any pair of strings $x$ and $y$ is equal to the length of the longest shared secret key that two parties can establish via a probabilistic protocol with interaction on a public channel, assuming that the parties hold as their inputs $x$ and  $y$ respectively. We determine the worst-case communication complexity of this problem  for the setting where the parties can use private sources of random bits.

We show that for some $x$, $y$ the communication complexity of the secret key agreement does not decrease even if the parties have to agree on a secret key the size of which is much smaller than the mutual information between $x$ and $y$. On the other hand, we discuss a natural class of $x$, $y$ such that the communication complexity of the protocol declines gradually with the size of the derived secret key. 

The proof of the main result uses the spectral properties of appropriate graphs and the expander mixing lemma, as well as information-theoretic techniques,
including constraint information inequalities and Muchnik's conditional descriptions.

\bigskip

A preliminary version of this paper was published in the proceedings of MFCS~2020. In the present version we give full 
proofs of all theorems and get rid of the assumption that the number of random bits used in the communication protocols 
is polynomial.
\end{abstract}

\maketitle

\section{Introduction}
			
In this paper we study communication protocols that help two remote parties (Alice and Bob) to establish a common secret key, while the communication is done via a public channel. In most practical situations, this task is achieved with the Diffie--Hellman scheme, \cite{diffie-hellman}, or other computationally secure protocols based on the assumptions that the eavesdropper has only limited computational resources and that some specific problem (e.g., the computing of the discrete logarithm) is computationally hard. We, in contrast, address the information-theoretic version of this problem and assume no computational restriction on the power of the eavesdropper. In this setting, the problem of the common secret key agreement can be resolved only if Alice and Bob possess since the very beginning some correlated data. In this setting, the challenge is to extract from the correlated data available to Alice and Bob their mutual information and materialize it as a common secret key. 
Besides the obvious theoretical interest, this setting is relevant for applications connected with quantum cryptography (see, e.g., \cite{devetak2005distillation,horodecki2009quantum}) or biometrics (see the survey \cite{ignatenko2012biometric}).

The problem of the secret key agreement was extensively studied in the classical information theory in the formalism of Shannon's entropy (this research direction dates back to the seminal papers \cite{Alswede-Csiszar-1993,Maurer-1993}). In our paper, we use the less common framework of algorithmic information theory based on Kolmogorov complexity. This approach seems to be an adequate language to discuss cryptographic security of an \emph{individual} key (while Shannon's approach helps to analyze properties of a probability distribution as a whole), see \cite{antunes2007cryptographic}.
Besides, Kolmogorov complexity provides a suitable framework for the \emph{one shot paradigm}, when we cannot assume that the initial data were obtained from an ergodic source of information or, moreover, there is possibly no clearly defined probability distribution on the data sources.

Thus, in this paper we deal with \emph{Kolmogorov complexity} and \emph{mutual information}, which are the central notions of algorithmic information theory. 
Kolmogorov complexity $C(x)$ of a string $x$ is the length of the shortest program that prints $x$. Similarly, Kolmogorov complexity $C(x|y)$ of a string $x$ given $y$ is the length of the shortest program that prints $x$ when $y$ is given as the input. 
Let us consider two strings $x$ and $y$. The mutual information $I(x:y)$ can be defined by the formula 
$I(x:y)=C(x)+C(y)-C(x,y)$. 
Intuitively, this quantity is the information shared by $x$ and $y$. The mutual information between $x$ and $y$ is equal (up to logarithmic additive terms) to the difference between the absolute and the conditional Kolmogorov complexity:
\[
I(x:y)\approx C(x)-C(x|y)\approx C(y)-C(y|x) 
\]
(see below the discussion of the Kolmogorov--Levin theorem). The ``physical'' meaning of the mutual information is elusive.
In general, the mutual information does not correspond to any common substring of $x$ and $y$, and we cannot materialize it as one object of complexity $I(x:y)$ that can be easily extracted from $x$ and separately from $y$. However, this quantity has a sort of \emph{operational interpretation}. The mutual information between $x$ and $y$ is essentially equal to the size of a longest shared secret key that two parties, one having $x$ and the other one having $y$, and both parties also possessing the complexity profile of the two strings can establish via a probabilistic protocol:	
\begin{theorem}[sketchy version; see \cite{RZ1} for a more precise statement]\label{th:rz}
(a) There is a secret key agreement protocol that, for every $n$-bit strings $x$ and $y$, allows Alice and Bob to compute with high probability a shared secret key $z$ of length equal to the mutual information of $x$ and $y$ (up to an $O (\log n)$ additive term).

(b) No protocol can produce a longer shared secret key (up to an $O (\log n)$ additive term).
\end{theorem}	
In this paper we study the communication complexity of the protocols that appear in this theorem.

The statement of Theorem~\ref{th:rz} in the form given above is vague and sketchy. Before we proceed with our results, we must clarify the setting of this theorem: we should explain the rules of the game between Alice, Bob, and the eavesdropper, and specify the notion of secrecy of the key in this context.

\smallskip
\noindent
\textbf{Clarification 1: secrecy.}
In this theorem we say that the obtained key $z$ is ``secret'' in the sense that it looks random. Technically, it must be (almost) incompressible, even from the point of view of the eavesdropper who does not know the inputs $x$ and $y$ but intercepts the communication between Alice and Bob. More formally, if $t$ denotes the transcript of the communication, we require that $C(z|t)\ge|z|-O(1)$. We will need to make this requirement even slightly stronger, see below.

If an $n$-bit  string  $z$ is incompressible in terms of Kolmogorov complexity,  i.e., if 
\[
C(z\mid \text{all publicly available information})\ge n - O(1),
\] 
this $z$ seems to be a suitable secret key for many standard cryptographic applications. Indeed, if any deterministic  algorithm $\mathcal A$ produces a list of strings, then a string $z$ whose complexity  is close to $n$ can appear in this list only at a position $\ge 2^{n-O(1)}$  
($z$ can be found given $\mathcal A$ and the index of $z$ in the list;  such a description should not  be shorter than $n - O(1)$, so the binary expansion of the index must be long enough).
 This means that any attempt to brute force this secret key will take $\Omega(2^n)$ steps before we arrive to the value $z$. Also,  it can be shown that for any computable distribution $\mathcal P$ 
\[
\prob[\text{a randomly sampled value is equal to } z] \le 2^{-n + O(\log n)}
\]
(see the relation between Kolmogorov complexity and the a priori probability, e.g., in \cite[Section~4.5]{shen-vereshchagin}),
which means that any attacker trying to  randomly guess the value of the secret key would need on average $\ge 2^{n-O(\log n)}$ false attempts before  $z$  is revealed.  
In other words, the standard properties of Kolmogorov complexity guarantee that any naive attempt to crack such a key (deterministically or with randomization) would take exponential time, so this
$z$ could be, for example,  a good code for a  safe with combination lock. 
A more systematic discussion of the foundations of cryptography in terms of Kolmogorov complexity can be found in  \cite{antunes2007cryptographic}. 

Incompressible strings are abundant. For every string $t$ and for every $n$, for a fraction $\ge 1-1/2^c$ of strings $z\in \{0,1\}^n$ we have
$
C(z|t) \ge n - c.
$
Thus, we can produce an incompressible $z$ by tossing an unbiased coin. However, Kolmogorov complexity is incomputable, and one cannot provide a certificate proving that some specific $z$ has large Kolmogorov complexity.

\smallskip
\noindent
\textbf{Clarification 2: randomized protocols.}
In our communication model we assume that Alice and Bob may use additional randomness. Each of them can toss a fair coin and produce a sequence of random bits with a uniform distribution. The private random bits produced by Alice and Bob are accessible only to Alice and Bob respectively. (Of course, Alice and Bob are allowed to  send the produced random bits to each other, but then this information becomes visible to the eavesdropper.) 

We assume that for each pair of inputs the protocol may \emph{fail} with a small probability (the probability is taken over the choice of random bits of Alice and Bob). The failure means that the parties do not produce any output, or Alice and Bob 
produce mismatched keys, or the produced key $z$ does not satisfy the secrecy condition.

In an alternative setting, Alice and Bob use a common \emph{public} source of randomness (also accessible to the eavesdropper). 
Theorem~\ref{th:rz} from \cite{RZ1} is valid for both settings --- for private or public sources of randomness,  though in the model with a public source of random bits the construction of a protocol  is 
much simpler.

\smallskip
\noindent
\textbf{Clarification 3: minor auxiliary inputs.} We assume also that besides the main inputs $x$ and $y$ Alice and Bob both are given the \emph{complexity profile} of the input, i.e., the values $C(x)$, $C(y)$, and $I(x:y)$. Such a concession is unavoidable for the positive part of the theorem. Indeed, Kolmogorov complexity and mutual information are non-computable; so there is no computable protocol that finds a $z$ of size $I(x:y)$ unless the value of the mutual information is given to Alice and Bob as a promise. 
This supplementary information is rather small, it can be represented by only $O(\log n)$ bits. The theorem remains valid if we assume that this auxiliary data is known to the eavesdropper. So, formally speaking, the protocol should find a key $z$ such that 
\[
C(z\mid \text{communication transcript of the protocol and complexity profile of $(x,y)$})\ge|z|-O(1).
\]

There exists no algorithm computing Kolmogorov complexity of an arbitrary string, the mutual information $I(x:y)$ cannot be algorithmically computed neither even given an access to both strings $x$ and $y$. However, the requirement  to provide  the complexity profile of  $(x,y)$ is not prohibitively unrealistic. Indeed,
if a string $x$ is taken uniformly at random from an explicitly given set $S_X$, we can claim that  $C(x)$ is close to $\log |S_X|$ with an overwhelming probability. Similarly, if a pair $(x,y)$ is taken uniformly from an explicitly given set of pairs $S_{XY}$, and if 
we know the combinatorial parameters of this set (the cardinality of the entire set  $S_{XY}$, of its projections on each of the coordinates, the sizes of the sections of this set along both coordinates) we can estimate the typical value of $I(x:y)$. Thus, we can predict an expected complexity profile for a ``typical'' pair $(x,y)\in S_{XY}$.

Even if  $(x,y)$ is produced by a more complicated random process, in many cases the expected  complexity profile of this pair  can be computed given the probabilistic characteristics of the distribution from which it is sampled 
(below we discuss in more detail the connection of Shannon's characteristics of a distribution and Kolmogorov complexity of a typical sampled outcome).

The communication protocol in Theorem~\ref{th:rz}(a)  still applies when Alice and Bob are provided with not the exact complexity profile of the input but only with
 its approximation. The more precise approximation we provide, the larger  secret key  is produced in the protocol, see  \cite[Remark~5]{RZ1}.
 For example, if Alice and Bob are given that $C(x) \ge k$ and $C(x|y)<\ell$, they can produce a secret key $z$ of size $k - \ell - O(\log (|x|+|y|))$. So,  if $k$ and $\ell$ are $\delta$-approximations
 for the real values of  $C(x)$ and $C(x|y)$, then the protocol results in a secret key of size $I(x:y) - O(\delta) -  O(\log (|x|+|y|))$.
  
Thus, if we  know the physical nature of the process producing the initial data for Alice and Bob, although we cannot predict exactly the complexity profile of the input, we possibly can estimate approximately the expected parameters of this profile. Such an approximation is enough to apply the communication protocol from  \cite{RZ1}.

\smallskip

Now we can formulate the main question studied in this paper:
\begin{question}
What is the optimal communication complexity of the communication problem from Theorem~\ref{th:rz}?
That is, how many bits should Alice and Bob  send to each other to  agree on a common secret key? 
\end{question}
 A protocol proposed in \cite{RZ1} allows  to compute  for \emph{all} pairs of inputs a shared secret key of length equal to the mutual information between $x$ and $y$ with communication complexity 
\begin{equation}\label{eq:comm-compl}
\min \{ C(x|y), C(y|x)\} + O(\log n).
\end{equation}
This protocol is randomized  (it can be adapted to the model with private or public sources of random bits).
Alice and Bob may need to send to each other different number of bits for different pairs of input (even with the same mutual information). It was proven in \cite{RZ1} that in the worst case (i.e., for \emph{some} pairs of inputs $(x,y)$) the communication complexity \eqref{eq:comm-compl} is optimal for communication protocols using only  \emph{public randomness}. The natural question whether this bound remains optimal for protocols with \emph{private} sources of random bits remained open (see \emph{Open Question~1} in \cite{RZ1}). The main result of this paper is the positive answer to this question. More specifically, we provide explicit examples of pairs $(x,y)$ such that 
\begin{equation}\label{eq:complexity-profile}
\left\{
\begin{array}{rcl}
I(x:y) &= & 0.5n + O(\log n)\\
C(x|y) &= & 0.5n + O(\log n)\\
C(y|x) &= & 0.5n + O(\log n)\\
\end{array}
\right.
\end{equation}
and in every communication protocol satisfying Theorem~\ref{th:rz} (with private random bits) Alice and Bob must exchange approximately $0.5n$ bits of information. Moreover, the same communication complexity is required even if Alice and Bob want to agree on a secret key that is much smaller then the optimal value $I(x:y)$. Our bound applies to all communication protocols that establish a secret key of size, say, $\omega(\log n)$.
 
\begin{theorem}\label{th:main}
Let $\pi$ be a communication protocol such that given inputs  $x$ and $y$ satisfying~\eqref{eq:complexity-profile}  Alice and Bob  use  private random bits and compute with probability $>1/2$ a shared secret key $z$ of length $ \delta(n) = \omega(\log n)$.
Then for every $n$ there exists a pair of $n$-bit strings $(x,y)$ satisfying~\eqref{eq:complexity-profile} such that following this communication protocol with inputs $x$ and $y$,  Alice and Bob send to each other  messages with a total length of at least $0.5n - O(\log n)$ bits. In other words, the worst-case communication complexity of the protocol is at least $0.5n - O(\log n)$.
\end{theorem}	
\begin{remark}
We assume that the computational protocol $\pi$ used by Alice and Bob is uniformly computable, i.e., the parties send messages and compute the final result by following rules that can be computed algorithmically given the length of the inputs. We require that the protocol can be applied to all pairs of inputs and all possible values of random bits. An outcome of the protocol can be meaningless if we apply it to pairs $(x,y)$ with too small mutual information or with an unlucky choice of random bits, but the protocol must always converge and produce some outcome. We may assume that the protocol is public (known to the eavesdropper). The constants hidden in the $O(\cdot)$ notation may depend on the protocol, as well as on the choice of the optimal description method in the definition of Kolmogorov complexity.

\label{uniform-protocols}
 An alternative approach might be as follows. We might assume that the protocol $\pi$ is not uniformly computable (but for each $n$ its description is available to Alice, Bob, and to the eavesdropper). Then substantially the same result can be proven  for Kolmogorov complexity relativized conditional on $\pi$. That is, we should define Kolmogorov complexity and mutual information in terms of programs that can access $\pi$ as an oracle, and the inputs $x$ and $y$ should satisfy a version of~\eqref{eq:complexity-profile} with the relativized Kolmogorov complexity.
 Our main result can be proven for this setting (literally the same argument applies). However, to simplify the notation, we focus on the setting with only computable communication protocols (the size of which does not depend on $n$).
 \end{remark}

Theorem~\ref{th:main} can be viewed as a special case of the general question of ``extractability'' of the mutual information studied in \cite{cmrsv}. We prove this theorem for two specific examples of pairs $(x,y)$. In the first example $x$ and $y$ are a \emph{line} and a \emph{point} incident with each other in a discrete affine plane. In the second example $x$ and $y$ are  points of the discrete plane with a fixed quasi-Euclidean distance between them. The proof of the main result consists in a combination of spectral and information-theoretic techniques. In fact, our argument applies to all pairs with similar spectral properties. Our main technical tools are the Expander Mixing Lemma (see Lemma~\ref{mixing-lemma}) and the lemma on non-negativity of the triple mutual information (see Lemma~\ref{l:triple-info}). We also use Muchnik's theorem on conditional descriptions with multiple conditions (see Proposition~\ref{p:muchnik}).

Let us mention one more technical detail specific for algorithmic information theory.
Many information-theoretical tools, including the Kolmogorov--Levin theorem, are inherently inaccurate: they involve an unavoidable logarithmic error term. This causes a difficulty if a communication protocol involves too many private random bits (then even a logarithmic term including the number of random bits can become overwhelming). To get around this obstacle, we must reduce the number of private random bits used in the protocol. Fortunately, using an argument similar to the classical Newman theorem, we can reduce the number of random bits involved in the protocol to a polynomial (or even a linear) function of the length of inputs (see Proposition~\ref{p:newman}).

The communication protocol proposed in \cite{RZ1} and Theorem~\ref{th:main}  imply together that we have the following threshold phenomenon. When the inputs given to Alice and Bob are  a line and a point (incident with each other in a discrete affine plane), then  the parties can agree on a secret key of size $I(x:y)$ with a communication complexity slightly \emph{above} $\min \{ C(x|y), C(y|x)\}$. But when a communication complexity is slightly \emph{below} this threshold, the optimal size of the secret key sinks immediately to $O(\log n)$.

We also show that the ``threshold phenomenon''  mentioned above is not universally true. There exist pairs $(x,y)$ with the same values of Kolmogorov complexities and the same mutual information as in the example above, but with a sharply different trade-off between the size of a secret key and the communication complexity needed to establish this key. In fact, for some $(x, y)$ the size of the optimal secret key decreases gradually with the communication complexity of the protocol. More specifically, we show that for some $x$ and $y$, Alice and Bob can agree on a secret key of any size $k$ (which can be chosen arbitrarily between $0$ and $I(x : y)$) via a protocol with a communication complexity of $\Theta(k)$. This is the case for $x$ and $y$ with an appropriate Hamming distance between them.

\smallskip

\emph{Secret key agreement in classical information theory.} The problem of \emph{secret key agreement} was initially proposed in the framework of classical information theory by Ahlswede and Csisz\'ar, \cite{Alswede-Csiszar-1993} and Maurer, \cite{Maurer-1993}. In these original papers the problem was studied for the case in which the input data is a pair of random variables $(X,Y)$ obtained by $n$ independent draws from a joint distribution (Alice can access $X$ and Bob can access $Y$). In this setting, the mutual information between $X$ and $Y$ and the secrecy of the key are measured in terms of Shannon entropy.
Ahlswede and Csisz\'ar in \cite{Alswede-Csiszar-1993} and Maurer in \cite{Maurer-1993} proved that the longest shared secret key that Alice and Bob can establish via a communication protocol  is equal to Shannon's mutual information between $X$ and $Y$.
This problem was extensively studied by many subsequent works in various restricted settings, see the survey \cite{sudan-survey}.
The  optimal communication complexity of this problem for the general setting remains unknown, though substantial progress has been made (see, e.g., \cite{Tyagi,Liu-Cuff-Verdu}).

There is a deep connection between the frameworks of classical information theory (based on Shannon entropy) and algorithmic information theory (based on Kolmogorov complexity),
see the detailed discussions in  \cite{li-vitanyi,shen-vereshchagin} and in  \cite{grunwald2004shannon}. It is know, for example, that for every computable distribution of probabilities 
on $\{0,1\}^*$ the average value of Kolmogorov complexity of a sampled string is close to Shannon entropy of the distribution, see \cite[Theorem~2.10]{grunwald2004shannon}.
In the simplest case, if $(X_i, Y_i)$ for $i=1,\ldots,n$ is a sequence of 
i.i.d. copies of  random pair  $(X,Y)$ distributed in some finite range, then with a high probability the pair of sampled values 
\[
\begin{array}{rcl}
x &\leftarrow& (X_1 \ldots  X_n),\\
y &\leftarrow&(Y_1 \ldots Y_n)
\end{array}
\]
has the profile of Kolmogorov complexities approximately proportional to the corresponding values of Shannon's entropy of $(X,Y)$,
\[
C(x) =  H(X) \cdot n + o(n), \ C(y) = H(Y) \cdot n + o(n), \ C(x,y) = H(X,Y)\cdot n + o(n),
\]
where $H(\cdot)$ denotes Shannon's entropy of a random variable. Using this observation one can transform a secret key agreement protocol in the sense of Theorem~\ref{th:rz}(a) in a secret key agreement 
for Shannon's settings (e.g., in the sense of \cite[Model S in Definition~2.1]{Alswede-Csiszar-1993} or  \cite[Definition~1]{Maurer-1993}).
We refer the reader to \cite{RZ1} for a more detailed discussion of parallels  between Shannon's and Kolmogorov's version of the problem of secret key agreement.
We believe that the correspondence between Shannon entropy and Kolmogorov complexity might help to use Theorem~\ref{th:rz} and Theorem~\ref{th:main}  
as a tool in the classical  information theory,  including (and especially) one-shot settings.

\smallskip

\emph{Differences between Shannon's and Kolmogorov's framework.}
Let us mention two important  distinctions between Shannon's and Kolmogorov's framework.
The first one  regards ergodicity of the input data. Most  results on secret key agreement in Shannon's framework are proven with the assumption that the input data  are obtained from a sequence of independent identically distributed random variables (or at least enjoy some  properties of ergodicity and stationarity). In the setting of Kolmogorov complexity we usually deal with inputs obtained in ``one shot''  without any assumption of ergodicity of the sources (see, in particular, Example~1 and Example~2 below).
Another distinction regards the definition of correctness of the protocol.
The usual paradigm in classical information theory is to require that the communication protocol works properly for \emph{most} randomly chosen inputs. In  Theorem~\ref{th:rz}  the protocol works properly with high probability for \emph{each} valid pair of input data (this approach  is more typical for the theory of communication complexity).

Though homologous statements in Shannon's and Kolmogorov's frameworks may look very similar, the proofs in these frameworks are usually pretty different.
Despite a close connection between Shannon entropy and Kolmogorov complexity,  many techniques from the classical information theory (designed for the framework of Shannon entropy) cannot be translated directly in the language of Kolmogorov complexity. 
Even manipulations with the chain rule for the mutual information and other standard information inequalities become more difficult in the framework of Kolmogorov complexity. The difficulty is that Kolmogorov's version of the chain rule is valid not absolutely but only up to an additive logarithmic term (see below in Preliminaries the comments on the Kolmogorov--Levin theorem). Such a minor error term can be harmless when it appears alone, but it causes difficulties when we need to sum up many copies of similar equalities or inequalities. (In particular, we deal with this difficulty in Section~4.2.)

\smallskip

\emph{Computational complexity of the protocol and potential practical realizations.}
 In this paragraph we discuss the relation of our settings to secret key agreement protocols that are, or eventually can be, deployed in practice. 
  As we argue above, the setting including non-computable pre-conditions (the complexity profile of the input data or its approximation given us as a promise) is not completely unrealistic. 
Indeed, in practice we may have partial knowledge of the distribution on the input data sets,
so that we might be able to estimate the expected value of Kolmogorov complexity profile for the inputs of the protocol 
(approximations of the numbers $C(x)$, $C(y)$, and $I(x:y)$). 

 The universal communication protocol constructed in \cite{RZ1} 
 is theoretically computable (once again, assuming that the necessary auxiliary data are given as a part of the input) 
 though it is not computationally efficient. 
This situation is frequently encountered in information theory.
For instance, the secret key agreement protocols constructed by Ahlswede-Csiszar  \cite{Alswede-Csiszar-1993} 
 and Maurer  \cite{Maurer-1993}   for Shannon's setting are also computationally non-efficient.
These results reveal the theoretically reachable limits; knowing such a limit, we can try to approach it with more efficient protocols (valid, perhaps,  in more restrictive settings).

We also observe that 
the connection between  Shannon entropy  and Kolmogorov complexity (see above) allows to 
transform a secret key agreement protocol in the setting of Theorem~\ref{th:rz}(a) in a secret key agreement 
for Shannon's settings (in the sense of  Model~S in \cite[Definition~2.1]{Alswede-Csiszar-1993}  or  \cite[Definition~1]{Maurer-1993}).
Thus, communication protocols valid for Kolmogorov's setting can be used  to provide a protocol with similar parameters  for the more classical Shannon's setting.

\smallskip
\emph{Plugging a secret key in a secure encryption scheme.}
It is typical for modern theoretical cryptography that fundamental concepts and protocols (cryptographic primitives) can be studied separately from each other.
Thus, in the literature on key agreement protocols,  and, more generally, on \emph{common randomness generation}, 
authors tend to focus on the problem of \emph{producing} the shared randomness and do not discuss potential \emph{usage} of this randomness.
In this paper we follow this tradition and ignore various applications of the secret keys resulted from such protocols. 
However, we keep in mind that in a typical application we need to achieve the secrecy not for a randomly sampled key \emph{per se} but for a message encrypted with such a key.
Of course, the conventional encryption schemes guarantee  that there is no leakage of information about the transmitted message assuming that  
the key used for encryption is kept in secret. 
This idea is pretty standard in the classical (probabilistic) settings. 
A similar framework of information-theoretic cryptography based on the notion of Kolmogorov complexity was proposed in \cite{antunes2007cryptographic}.  
In Section~\ref{s:antunes-theory} we briefly recall the basic definitions of security used in this approach.

\smallskip

\emph{Overview of the proof of the main result.}
Let us sketch the proof of Theorem~\ref{th:main}, the main result of this paper.
 For every large enough $n$ we provide  a hard pair of inputs $(x,y)$ for which the communication complexity must be large.
To this end we construct explicitly a bipartite graph  $G_n$ with $N= 2^n$ vertices in each part and with degree of each vertex $D = 2^{n/2}$. 
Most edges $(x, y)$ in such a graph  satisfy the property that  $C(x)\eqp C(y)\eqp n$ 
and $I(x : y) \eqp n/2$.   Theorem~\ref{th:rz} claims  that for these $x$ and $y$ Alice and Bob can produce a secret key of length  approximately  $I(x:y)$ by exchanging approximately $C(x|y)$ bits. 
We prove  that any protocol will have to exchange for these pairs of inputs $(x,y$) approximately the same number of bits, even to agree on  a much shorter secret key.

We fix an arbitrary communication protocol $\pi$, take a typical edge $(x,y)$  in $G_n$ and proceed as follows. 
Let $z=z(x,y)$ denote the secret key produced by $\pi$ for inputs $(x,y)$ and $t=t(x,y)$ denote the transcript (sequence of messages sent by Alice and Bob following the protocol).
Our aim is to show  that $t$ consists of at least $n/2 - O(\log n)$ bits.

On the one hand, we show that the mutual information  between $x$ and $y$ can only decrease if we add $t$ as a condition, i.e., 
$
I(x:y|t) \lep I(x:y),
$
and it must decrease further by at least $|z|$ bits if we add to the condition $z$. Thus, the gap between $I(x:y|t,z)$ and $I(x:y)$ must be at least $|z|$.  In this part of the  argument
we have to overcome the technical difficulty mentioned  above: we cannot use freely the chain rule if the number of rounds in the protocol is unbounded
(to handle this difficulty, we use the lemma from \cite{RZ1} saying that  the \emph{external information complexity} cannot be less than  the \emph{internal information complexity}, and then 
apply Muchnik's theorem on conditional descriptions to reduce a generic communication protocol to the case where the external and  the internal information complexity are equal
to each other).

On the other hand, we use the expander mixing lemma to show that the gap between $I(x:y|w)$ and $I(x:y)$ is negligibly small for every $w$ with a small enough  Kolmogorov complexity (if $C(w)$ is not greater than $0.5n$).  This leads to a contradiction  
if complexity of  $w = \langle t,z\rangle$ is smaller than $0.5n$. By reducing the key size $|z|$ to $\Theta(\log n)$ (which only makes the protocol weaker) we conclude that the contradiction can be avoided
only if $C(t)$ alone is greater than $0.5n - O(\log n)$. This means that the transcript must contain at least $0.5n - O(\log n)$ bits, which concludes the proof.
This part of the argument works if the graph $G_n$ has a large spectral gap  (a large difference between the first and the second eigenvalue). This condition is necessary to apply the expander mixing lemma.

In this sketch, we did not  take into account the random bits produced by Alice and Bob.
Adding private random bits to the overall picture complicates the technical details of the argument, but does not lead to any conceptual difficulties.

\smallskip

\emph{The rest of the paper is organized as follows.}
In Preliminaries (Section~2) we sketch the basic definitions and notations for Kolmogorov complexity and communication complexity. 
In Section~3 we translate information  theoretic properties of pairs $(x,y)$ in the language of graph theory and present three explicit examples of pairs $(x,y)$ satisfying \eqref{eq:complexity-profile}.
\begin{itemize} 
\item Example~1 involves finite geometry, $x$ and $y$ are incident points and lines on a finite plane; 
\item Example~2 uses a discrete version of the Euclidean distance, $x$ and $y$ are points on the discrete plane with a known quasi-Euclidean distance between them; 
\item Example~3 involves binary strings  $x$ and $y$ with  a fixed Hamming distance between them. 
\end{itemize}
The pairs $(x,y)$ from these examples have pretty much the same complexity profile, but the third example has significantly different spectral properties.
The proof of the main result sketched above applies only to Example~1 and Example~2 (for Example~3 not only the \emph{argument} fails but the \emph{statement} of the main result is false\footnote{%
The revealed difference between Examples~1~and~2 on the one hand and Example~3 on the other hand 
shows that the optimal communication between Alice and Bob is not determined by the complexity profile of the input and depends on subtler properties of $(x, y)$.}).

In Section~4 we use a spectral technique to analyze the combinatorial properties of graphs.
We combine spectral bounds with information-theoretic arguments and prove our main result (Theorem~\ref{th:main}) for the pairs $(x,y)$ from Example~1 and Example~2 mentioned above.

In Section~5 we show that the statement of Theorem~\ref{th:main} is not true for the pairs $(x,y)$ from our Example~3: for those $x$ and $y$ there is no ``threshold phenomenon'' mentioned above, and the size of the longest achievable secret key depends continuously on the communication complexity of the protocol, see Theorem~\ref{th:hamming1} and Theorem~\ref{th:hamming2}.

In Appendix we provide a self-contained proof of a version of Newman's theorem on the reduction of the number of random bits used in a communication protocol 
(we prove a variant of this theorem for protocols with independent sources of private randomness).%

\section{Preliminaries}

\paragraph{Kolmogorov Complexity.}
Given a Turing machine $M$ with two input tapes and one output tape, we say that $p$ is a program that prints a string $x$ conditional on $y$ (a description of $x$ conditional on $y$) if $M$ prints $x$ on the pair of inputs $p$, $y$. Here $M$ can be understood as an interpreter of some programming language that simulates a program $p$ on a given input $y$.
We denote the length of a binary string $p$ by $|p|$. The \emph{algorithmic complexity} of $x$ \emph{conditional on} $y$ relative to $M$ is defined as 
\[
C_M(x|y)=\min\{|p|:M(p,y)=x\}.
\]
It is known that there exists an \emph{optimal} Turing machine $U$ such that for every other Turing machine $M$ there is a number $c_M$ such that for all $x$ and $y$
\[
C_U(x|y)\le C_M(x|y)+c_M.
\]
Thus, if we ignore the additive constant $c_M$, the algorithmic complexity of $x$ relative to $U$ is minimal. In the rest we fix an optimal machine $U$, omit the subscript $U$ and denote
\[
C(x|y):=C_U(x|y).
\]
This value is called \emph{Kolmogorov complexity} of $x$ conditional on $y$. Kolmogorov complexity of a string $x$ is defined as the Kolmogorov complexity of $x$ conditional on the empty string $\Lambda$,
\[
C(x):=C(x|\Lambda).
\]
We fix an arbitrary computable bijection between binary strings and all finite tuples of binary strings (i.e., each tuple is encoded in a single string). Kolmogorov complexity of a tuple $\langle x_1,\ldots,x_k\rangle$ is defined as Kolmogorov complexity of the code of this tuple. For brevity we denote this complexity by $C(x_1,\ldots,x_k)$.

We use the conventional notation
\[
I(x:y) := C(x) +C(y) - C (x,y)
\]
and 
\[
I(x:y|z) := C(x|z) +C(y|z) - C (x,y|z).
\]
In this paper we  use systematically the Kolmogorov--Levin theorem, \cite{zvonkin-levin}, which claims that all $x,y$.
\[
| C(x|y) + C(y) - C(x,y) | = O(\log(|x|+|y|)).
\]
It follows, in particular, that
\[
I(x:y) = C(x) - C(x|y) + O(\log(|x|+|y|)) = C(y) - C(y|x) + O(\log(|x|+|y|)). 
\]

Since many natural equalities and inequalities for Kolmogorov complexity hold up to a logarithmic term, we abbreviate some  formulas by using  the notation
$
A\eqp B,\ A\lep B, \text{ and } A\gep B
$
for 
\[
|A-B| = O(\log n),\ A\le B + O(\log n), \text{ and } B\le A + O(\log n)
\]
respectively, where $n$ is clear from the context. Usually $n$ is the sum of the lengths of all strings involved in the inequality.
In particular, the Kolmogorov--Levin theorem can be rewritten as  $C(x,y) \eqp C(x|y) + C(y)$.

We also use the notation
\[
I(x:y:z) := C(x) + C(y) + C(z) - C(x,y) - C(x,z) - C(y,z) + C(x,y,z).
\]
Using the Kolmogorov--Levin theorem it is not hard  to show that
\[
 I(x:y:z)  \eqp  I(x:y) - I(x:y|z)
\eqp  I(x:z) - I(x:z|y) 
 \eqp  I(y:z) - I(y:z|x).
\]
These relations can be observed on a Venn-like diagram, as shown in Fig.~\ref{x-y-z}.
			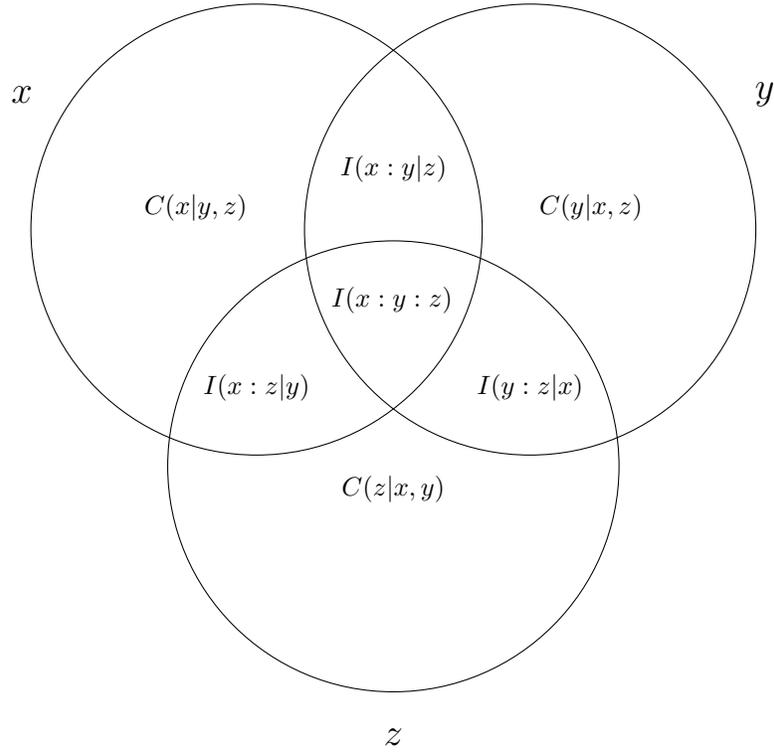
\begin{figure}
				\centering
				\begin{tikzpicture}[scale=1.0]
				  \draw \firstcircle node[above left] {$C(x|y,z)$};
				  \draw \secondcircle node [above right] {$C(y|x,z)$};
				  \draw \thirdcircle node [below] {$C(z|x,y)$};
				  \node at (100:0.10)   {$I(x:y:z)$};
				  \node at (90:1.55) {$ I(x:y|z)$};
				  \node at (210:1.75) {$I(x:z|y)$};
				  \node at (330:1.75) {$I(y:z|x)$};
				  \node at (150:4.75) {\Huge $x$};
				  \node at (30:4.75) {\Huge $y$};
				  \node at (270:4.75) {\Huge $z$};
				\end{tikzpicture}
				\caption{Complexity profile for a triple $x,y,z$. On this diagram it is easy to observe several standard equations: 
				\newline
				$\bullet$ $C(x) = C(x|y,z) +  I(x:y|z) +  I(x:z|y) + I(x:y:z)$ 
				\newline
				\rule{3mm}{0mm}(the sum of all quantities inside the left circle representing $x$);
				\newline
				$\bullet$ $I(x:y) = I(x:y|z) + I(x:y:z)$ 
				\newline
				\rule{3mm}{0mm}(the sum of the quantities in the intersection of the left and the right circles
				\rule{3mm}{0mm}representing $x$ and $y$ respectively);
				\newline
				$\bullet$ $C(x,y) = C(x|y,z) + C(y|x,z) +  I(x:y|z) +  I(x:z|y)  +  I(y:z|x) + I(x:y:z)$ 
				\newline
				\rule{3mm}{0mm}(the sum of all quantities inside the union of the left and the right circles);
				\newline
				$\bullet$ $C(x|y) = C(x|y,z) +  I(x:z|y) $ 
				\newline
				\rule{3mm}{0mm}(the sum of the quantities inside the left circle but outside the right one);
				\newline
				and so on; all these equations are valid up to $O(\log(|x|+|y|+|z|))$.
				}\label{x-y-z}
			\end{figure}
			
In the usual jargon, $x$ is said to be (almost) \emph{incompressible} given $y$ if
\[
C(x|y) \gep |x| ,
\]
and $x$ and $y$ are said to be \emph{independent}, if $I(x:y)\eqp0$.

For a survey of basic properties of Kolmogorov complexity we refer the reader to the introductory chapters of \cite{li-vitanyi} and \cite{shen-vereshchagin}. 

\paragraph{Communication Complexity.} 
In what follows we use the conventional notion of a communication protocol for two parties (traditionally called Alice and Bob), see for detailed definitions \cite{kushilevitz}.
In a deterministic communication protocol the inputs of Alice and Bob are denoted $x$ and $y$ respectively. A deterministic communication protocol is said to be correct for inputs of length $n$ if for all $x,y\in\{0,1\}^n$, following this protocol Alice and Bob obtain a valid result $z=z(x,y)$. The sequence of messages sent by Alice and Bob to each other while following the steps of the protocol is called a \emph{transcript} of the communication.

In the setting of randomized protocols with \emph{private} sources of randomness, Alice can access $x$ and an additional string of bits $r_A$, and Bob can access $y$ and an additional string of bits $r_B$. A randomized communication protocol is said to be correct for inputs of length $n$ if for \emph{all} $x,y\in\{0,1\}^n$ and for \emph{most} $r_A$ and $r_B$, following this protocol Alice and Bob obtain a valid result $z=z(x,r_A,y,r_B)$.

In a secret key agreement protocol, \emph{correctness} of the result $z$ means that (i)~$z$ is of the required size and (ii)~it is almost incompressible even given the transcript of the communication. That is, if $t=t(x,r_A,y,r_B)$ denotes the transcript of the communication, then $C(z|t)$ must be close enough to $|z|$. Note that in this setting the transcript $t$ and the final result $z$ are not necessarily functions of the inputs $x,y$, they may depend also on the random bits used by Alice and Bob. For a more detailed discussion of this setting we refer the reader to \cite{RZ1}.

Communication complexity of a protocol is the maximal length of its transcript, i.e., $\max\limits_{x,r_A,y,r_B} |t(x,r_A,y,r_B)|$.

In this paper we deal with communication protocols that can be meaningfully applied only to pairs of inputs with a special property
(in particular, the mutual information between the inputs of Alice and Bob must be large enough). However, we assume that by definition a communication protocol  can be formally applied to any inputs  with any values of random bits. This means that with any input data
a protocol produces some communication transcript, and each party ends up with some outcome.

\label{par:uniform-protocol}
In communication complexity, it is typical to study a communication protocol as a scheme defined for inputs of one fixed length. 
In this paper, we prefer to deal with protocols that are defined for inputs of all lengths. We always assume that the communication protocols under consideration are \emph{uniformly computable}: there are algorithms for Alice and Bob that compute the messages to be sent by the parties, and produce at the end of the protocol the final result. Uniformity  means that the same algorithms can be used for inputs of all lengths.
This approach simplifies the formal definition of secrecy --- we may assume that the communication protocol 
(an algorithm that computes this protocol) is a finite object, and this object is known to Alice, Bob, and to the eavesdropper
(see also Remark~\ref{uniform-protocols} on p.~\pageref{uniform-protocols}). 
We may ignore the description of a protocol in the definition of secrecy, since the information on a finite object affects only additive constant terms in all expressions with Kolmogorov complexity.

\paragraph{The Amount of Randomness in Communication Protocols.} In the usual definition of a randomized communication protocol, Alice and Bob are allowed to use any number of random bits. However, for most settings, it can be shown that the amount of randomness in use can be reduced without loss of the performance of the protocol. It is known that for a communication protocol  computing a function of Alice' and Bob's inputs with help of a \emph{public} source of randomness, the number of used random bits can be reduced to $O(\log n)$ (where $n$ is an upper bound on the length of the inputs of the protocol) 
at the price of only a minor degradation of the error probability, see Newman's theorem \cite{newman1991private}. We cannot apply Newman's theorem directly in our setting. 
The main reason is that we care about not only the result obtained by Alice and Bob but also about the relation between the final result (the secret key) and the transcript of the communication (the secrecy condition). In general, 
we cannot reduce the number of random bits used  in such a protocol to  $O(\log n)$, but we will see that  $O(n)$ random bits is always enough. 
In an arbitrary protocol of this type we cannot delegate the task of generating the random bits to only one participant of the protocol, and Alice and Bob must be careful when communicating the produced random bits to each other. Indeed, sending the random bits via the communication channel not only affects the communication complexity of the protocol but may also reveal to the eavesdropper some information on the resulting secret key.
Thus, another technical difficulty with our communication protocols is that  we have to deal with two \emph{private} sources of randomness (we need to handle separately the random bits produced by Alice and the ones produced by Bob). 
However, the ideas used in the proof of Newman's theorem (random sampling) can be adapted to this setting. We are not aware of such results being previously published, so in what follows we present it in more detail.

We deal with two-party communication protocols where Alice and Bob are given inputs $x$ and $y$ respectively (usually of the same length $n$) and can use any number of private random bits. Following the protocol, the parties exchange several messages (this sequence of messages is a transcript $t$ of the protocol), and end up with some results --- final answers $z_1$ (the answer of Alice) and $z_2$ (the answer of Bob). 
For every pair of inputs $(x,y)$, a randomized communication protocol can produce with different probabilities different answers $z_1,z_2$ and different transcripts $t$. (We usually require that Alice and Bob obtain with a high probability the same answer, i.e., $z_1=z_2$. However, in general the results obtained by two parties can be different.) These outcomes of the protocol can be \emph{valid} or \emph{invalid}. In the classical settings, where Alice and Bob compute some function $F(x,y)$, the results are valid if $z_1=z_2=F(x,y)$. 
In the setting where Alice and Bob compute a relation $R$, an answer $z_i$ is valid if the relation $R(x,y,z_i)$ is true. In a secret key agreement protocol the condition of validity is more involved: for each admissible pair of inputs $(x,y)$ (technically, for all $x$ and $y$ satisfying~\eqref{eq:complexity-profile}) the outcome of the protocol is valid if
\begin{itemize}
\item answers $z_1$ and $z_2$ are equal to each other;
\item Kolmogorov complexity of the obtained key $z_1=z_2$ is bigger than some fixed threshold
(e.g., we may require that Kolmogorov complexity of the key is close to $I(x:y)$);
\item the information in the transcript  on the obtained secret key is smaller than some fixed threshold
(e.g., we may require that $C(z_i|t) \ge |z_i| - O(\log n)$).
\end{itemize}
Let $\pi$ be a  two-party communication protocol of secret key agreement
with private randomness such that  for all inputs $x$ and $y$ satisfying~\eqref{eq:complexity-profile} 
with a probability $>1-\varepsilon$  the produced answers $z_1,z_2$ and the transcript $t$ are valid in the sense explained above. 
We may assume that the transcript is reasonably short (its length is at most $O(n)$, as it is the case for all protocols of interest).
Under this assumption, we argue that there is another communication protocol $\pi'$ that produces valid outcomes
with almost the same probability and uses only $O(n)$ random bits.

\begin{proposition}\label{p:newman}
Let $\varepsilon_1,\varepsilon_2$ be positive numbers, and let 
 $\pi$ be a  two-party communication protocol of secret key agreement
with private randomness such that  for all inputs $x$ and $y$ satisfying~\eqref{eq:complexity-profile} 
with a probability $>1-\varepsilon_1$  the produced answers $z_1,z_2$ and the transcript $t$ are valid in the sense explained above. 
Assume that the lengths of the produced answers $z_1$ and $z_2$ is always bounded by $O(n)$ and
communication complexity of the protocol is $O(n)$.
Then there exists a communication protocol $\pi'$ that produces valid outcomes
with a probability $>1-\varepsilon_1-\varepsilon_2$, and Alice and Bob use at most $O(n+\log(1/\varepsilon_2))$ private random bits.

Moreover, if the original protocol was uniformly computable (there is an algorithm that implements the protocol for all inputs of all lengths), then the new protocol also can be made uniformly computable.
\end{proposition}
Notice that we do not assume that the notion of \emph{validity} of  the  output of a protocol is computable (we cannot assume that because Kolmogorov complexity and the mutual information are not computable). However, we do assume that the protocol is defined for all pairs of inputs (even on those which do not satisfy~\eqref{eq:complexity-profile} and for which  the outcome of the communication is meaningless).
The proof of Proposition~\ref{p:newman} uses the standard idea of random sampling. The proof is deferred to Appendix.

\section{From Information-Theoretic Properties to Combinatorics of Graphs}

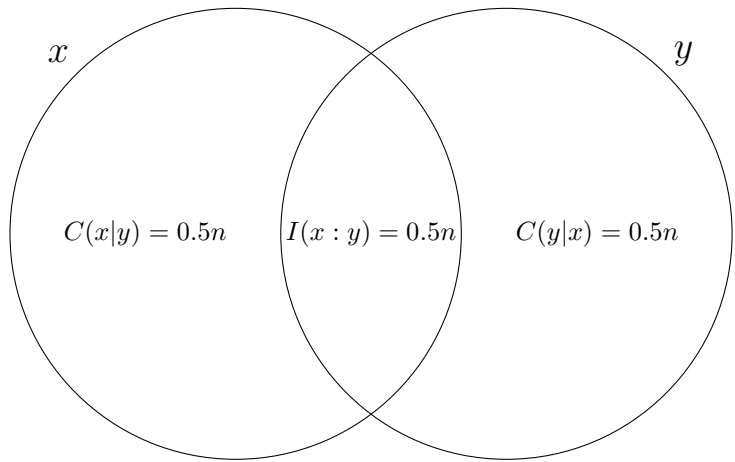
\begin{figure}
				\centering
				\begin{tikzpicture}[scale=1.10]
				  \draw \leftcircle node[left] {$C(x|y)=0.5n$};
				  \draw \rightcircle node [right] {$C(y|x)=0.5n$};
				  \node {$I(x:y)=0.5n$};
				  \node at (150:4) {\Huge $x$};
				  \node at (30:4) {\Huge $y$};
				\end{tikzpicture}
				\caption{A diagram for the complexity profile of two strings $x,y$: from the Kolmogorov--Levin theorem we have $C(x) \eqp C(x|y)+ I(x:y) \eqp n$, $C(y) \eqp C(y|x)+ I(x:y) \eqp n$, and $C(x,y)\eqp C(x|y)+ C(y|x)+I(x:y)\eqp  1.5n.$}\label{x-y-venn-diagram}
			
\end{figure}

To study the information-theoretic properties of a pair $(x,y)$ we will embed this pair of strings in a large set of pairs that are in some sense similar to each other. We will do it in the language of bipartite graphs. The information-theoretic properties of the initial pair $(x,y)$ will be determined by the combinatorial properties of these graphs. In their turn, the combinatorial properties of these graphs will be proven using the spectral technique.
In this section we present three examples of $(x,y)$ corresponding to three different constructions of graphs. In the next sections we will study the spectral and combinatorial properties of these graphs and, accordingly, the information-theoretic properties of these pairs $(x,y)$.

We start with a simple lemma that establishes a correspondence between information-theoretic and combinatorial language for the properties of pairs $(x,y)$.
\begin{lemma}\label{l:graph-complexity-profile}
Let $G=(L\cup R,E)$ be a bipartite graph such that $|L|=|R|=2^{n+O(1)}$ and the degree of each vertex is $D=2^{0.5n +O(\log n)}$. We assume that this graph  has an \emph{explicit construction} in the sense that the complete description of this graph (its adjacency matrix) has Kolmogorov complexity $O(\log n)$.
Then most\footnote{For every $d>0$, we can make the fraction of such pairs greater than  is $1-1/n^{d}$ by choosing the constants hidden in the terms $O(\log n)$ in \eqref{eq:graph-complexity-profile}.} $(x,y)\in E$ (pairs of vertices connected by an edge) have the following complexity profile:
\begin{equation}\label{eq:graph-complexity-profile}
\left\{
\begin{array}{rcccl}
 n - O(\log n)  &\le& C(x) &\le& n + O(\log n),\\
n - O(\log n) &\le& C(y) &\le& n + O(\log n),\\
1.5n - O(\log n) &\le& C(x,y) &\le& 1.5n + O(\log n)
\end{array}
\right.
\end{equation}
(which is equivalent to the triple of equalities in \eqref{eq:complexity-profile}, see Fig.~\ref{x-y-venn-diagram}).
\end{lemma}
\begin{proof}
There are $D\cdot |L|  = D\cdot |R| = 2^{1.5n+O(\log n)}$ edges in the graph. Each of them can be specified by its index in the list of elements of $E$,
and this index should consists of only $\log|E|$ bits;  the set $E$ itself can be described by $O(\log n)$ bits. Therefore, 
\[
C(x,y) \le \log|E| + O(\log n) = 1.5n+O(\log n).
\]
For most $(x,y)\in E$ this bound is tight. Indeed, for every number $c$ there are at most $2^{1.5n-c}$ objects with complexity less than $1.5n-c$.
Therefore,  for all but $|E|/2^{O(\log n)}= 2^{1.5n-O(\log n)}$ edges in $E$ we have 
\begin{equation}\label{eq:C(xy)>1.5n-O(logn)}
  C(x,y)  \ge 1.5n-O(\log n).
\end{equation}
Similarly, for each $x\in L$ and for each $y\in R$ we have
\[
C(x)\le\log|L|+O(\log n)=n+O(\log n),\ C(y)\le\log|R|+O(\log n)=n+O(\log n).
\]
To prove that for most pairs $(x,y)$ these bounds are also tight, we use the fact that each vertex in the graph has $D$ neighbors. 
A pair $(x,y)\in E$ can be specified by a description of $x$ combined with the index of $y$ in the list of all neighbors of $x$. 
It follows that 
\[
C(x,y) \le C(x) + \log D + O(\log n) = C(x) + 0.5 n + O(\log n).
\]
 Therefore, if $C(x)$ is significantly (by at least $\omega(\log n)$) less than $n$, then $C(x,y)$ must be significantly less than $1.5n$. But we have seen that \eqref{eq:C(xy)>1.5n-O(logn)} is true  for the vast majority of pairs $(x,y)$. 
Hence, for the vast majority of pairs we have $ C(x) \ge n - O(\log n).$
A similar argument implies that for most pairs $(x,y)\in E$  we have $C(y) \ge n - O(\log n)$.
\end{proof}
\begin{remark}
In a graph satisfying the conditions of Lemma~\ref{l:graph-complexity-profile} each vertex has $D=2^{0.5n+O(\log n)}$ neighbors. Therefore, for \emph{all} $(x,y)\in E$ we have 
\[
C(x|y)  \le 0.5 n+O(\log n),\  C(y|x) \le  0.5 n+O(\log n),
\]
(given $x$, we can specify $y$ by its index in the list of all neighbors of $x$ and vice-versa.)
From Lemma~\ref{l:graph-complexity-profile} and the Kolmogorov--Levin theorem it follows that for \emph{most}  $(x,y)\in E$ these inequalities are tight, i.e.,  $C(x|y) = 0.5n \pm O(\log n)$ and $C(y|x) = 0.5n \pm O(\log n)$.
\end{remark}
\begin{remark}
The density of edges in the graph $G=(L\cup R,E)$ (i.e., the ratio $\frac{|E|}{|L|\cdot|R|}$) corresponds on the logarithmic scale to the mutual information between $x$ and $y$. Indeed,  
the equations in \eqref{eq:graph-complexity-profile} mean that for most $(x,y)\in E$
		\[
		\frac{|E|}{|L|\cdot|R|}=\frac{2^{C(x,y)+O(\log n)}}{2^{C(x)+O(\log n)}\cdot2^{C(y)+O(\log n)}}=2^{-I(x:y)\pm O(\log n)}.
		\]
Below we pay attention to the density of edges in the \emph{induced subgraphs} of $G$. We will see that this ratio corresponds in some sense to the property of ``extractability'' of the mutual information. We will show that for some specific graphs $G$ satisfying Lemma~\ref{l:graph-complexity-profile}, in all large enough induced subgraphs, the density of edges is close to $2^{-I(x:y)}$.
\end{remark}

\begin{example}[discrete plane]\label{lines-and-points}
Let $\mathbb{F}_q$ be a finite field of cardinality $q=2^n$. Consider the set $L$ of points on plane $\mathbb{F}_q^2$ and the set $R$ of non-vertical lines, which can be represented as affine functions $y=ax-b$ for $(a,b)\in \mathbb{F}_q^2$. Let $G=(L\cup R,E)$ be the bipartite graph where a point $(x_0,y_0)$ is connected to a line $y=ax-b$ if and only if it is on the line i.e. $y_0=ax_0-b$. Clearly $|L|=|R|=2^{2n}$. The degree of each vertex is $2^n$ since there are exactly $q$ points on each line and there are exactly $q$ lines on each point. 
In the sequel we denote this graph by $G^{Pl}_n$.

This graph (or its adjacency matrix) can be constructed effectively when the field $\mathbb{F}_q$ is given. We assume a standard construction of the field $\mathbb{F}_{2^n}$ to be fixed.  Thus,  the graph is uniquely defined by the binary representation of $n$. Therefore, we need only $O(\log n)$ bits  to describe the graph (as a finite object).  
Lemma~\ref{l:graph-complexity-profile} applies to this graph, so for most $(x,y)\in E$ the equalities in  \eqref{eq:graph-complexity-profile} are satisfied.
\end{example}
The graph from Example~\ref{lines-and-points} was used in \cite{muchnik1998common,cmrsv} to construct an  example of a pair with non-extractable mutual information. We will use a similar property of  this graph (although formulated  differently and proven by a different technique, see Remark~\ref{r:lines-points-history} below).

\begin{example}[discrete Euclidean distance]\label{quasi-distance}
Let $\mathbb{F}_q$ be a finite field of order $q$, where $q$ is an odd prime power. Let us define the distance function between two points $x=(x_1,x_2)$, $y=(y_1,y_2)$ in $\mathbb{F}_q^2$ as
\[
\mathrm{dist}(x,y)=(x_1-y_1)^2+(x_2-y_2)^2.
\]
For every $r\in\mathbb{F}_q\setminus\{0\}$ we define the \emph{finite Euclidean distance graph} $G=(L\cup R,E)$ as follows:
$
L=R=\mathbb{F}_q^2,
$
and 
\[
E=\{(x,y))\ \mid\ \mathrm{dist}(x,y)=r\}.
\]
Obviously, $|L|=|R|=q^2$. 

This graph  is regular: for each point $x\in \mathbb{F}_q^2$ the number of neighbors is equal  by definition to the number of  points  $y=x+u$ where the coordinates of $u=(u_1,u_2)$ 
satisfy $u_1^2+u_2^2 = r$. This number depends on $q$ and $r$ but not on a specific $x$. It can be shown that for all $r\not=0$ the degree of this graph is $\Theta(q)$, and $|E|=\Theta(q^3)$, see \cite{quasi-dist}.

For definiteness, we can fix for each $n$ one $q_n$ such that $\lceil\log (q_n^2)\rceil=n$.
Since the prime numbers are dense enough\footnote{It is known, see  \cite{nagura1952interval}, that for all $x\ge25$ there exists a prime number between $x$ and $(1+\frac15)x$. Since $(1+\frac15)^2<2$, this implies that for all $n\ge10$ we can find a square of a prime number  between $2^{n-1}$ and $2^{n}$ (starting, for example, with the number $23$ whose square lies between $2^9$ and $2^{10}$).}, for every large enough  integer $n$ we can fix a suitable prime number $q_n$. 

For the defined above graph $G=(L\cup R,E)$ for this $\mathbb{F}_{q_n}$ we have $|L|=|R|=2^{n+O(1)}$ and $|E|=2^{1.5n+O(1)}$, and Lemma~\ref{l:graph-complexity-profile} applies to this graph. We should also fix the value of $r$. Any non-zero element of   $\mathbb{F}_{q_n}$ would serve the purpose. 
For simplicity we may assume that $r=1$.
Thereafter we denote this graph by $G^{Euc}_n$.
\end{example}

In our next example we use the following standard lemma.
\begin{lemma}\label{l:binom}
	Denote $h(t):= -t\log t-(1-t)\log(1-t)$. For any real number $\gamma\in(0,1)$ and every positive integer $m$, ${\binom m{\gamma m}}=2^{h(\gamma)m+O(\log m)}$.
\end{lemma}

\begin{example}[Hamming distance]\label{hamming}
We choose $\theta\in(0,\frac12)$ such that $h(\theta)=0.5$. Let $L=R=\{0,1\}^n$. We define the bipartite graph $G=(L\cup R,E)$ so that two strings (vertices) from $L$ and $R$ are connected if and only if the Hamming distance between them is $\theta n$. Clearly $|L|=|R|=2^n$. By Lemma~\ref{l:binom} the degree of each vertex is $D=\binom n{\theta n}=2^{0.5n+O(\log n)}$.
Lemma~\ref{l:graph-complexity-profile} applies to this graph. Therefore, for most $(x,y)\in E$ we have \eqref{eq:graph-complexity-profile}. Subsequently we denote this graph by $G^{Ham}_{\theta,n}$.
\end{example}

We are interested in properties of $(x,y)$ that are much subtler than those from Lemma~\ref{l:graph-complexity-profile}. For example, we would like to know whether there exists a $z$ materializing a part of the mutual information between $x$ and $y$ (i.e., such that $C(z|x)\approx0$, $C(z|y)\approx0$, and $C(z)\gg0$). These subtler properties are not determined completely by the ``complexity profile'' of $(x,y)$ shown in Fig.~\ref{x-y-venn-diagram}. In particular, we will see that some of these properties are different for pairs $(x,y)$ from Example~\ref{lines-and-points} and Example~\ref{quasi-distance} on the one hand and from Example~\ref{hamming} on the other hand.
In the next section we will show that some information-theoretic properties of $(x,y)$ are connected with the spectral properties of these graphs.

\paragraph{Randomized Communication Protocols in the Information-Theoretic Framework.}
In our main results we discuss communication protocols with two parties, Alice and Bob, who are given inputs $x$ and $y$. We will assume that  Alice and Bob are given the ends of some edge $(x,y)$ from $G^{Pl}_n$, from $G^{Euc}_n$,  or from  $G^{Ham}_{\theta,n}$. 

We admit randomized communication protocols with private sources of randomness. Technically this means that besides the inputs $x$ and $y$, Alice and Bob are given strings of random bits, $r_A$ and $r_B$ respectively. We assume that both $r_A$ and $r_B$  are binary strings from $\{0,1\}^m$. Without loss of generality   we may assume that $m=\poly(n)$, see Proposition~\ref{p:newman}.

It is helpful to represent the entire inputs available to Alice and Bob as an edge in a graph. The data available to Alice are $x':=\langle x,r_A\rangle$ and the data available to Bob are $y':=\langle y,r_B\rangle$. We can think of the pair $(x',y')$ as an edge in the graph
\[
\widehat{G^{Pl}_n} : = G^{Pl}_n \otimes K_{M,M}
\]
\label{hat-g-pl}
(if $(x,y)$ is an edge in $G^{Pl}_n$), or
\[
\widehat{G^{Euc}_n} : = G^{Euc}_n \otimes K_{M,M}
\]
(if $(x,y)$ is an edge in $G^{Euc}_n$), or
\[
\widehat{G^{Ham}_{\theta,n}} : = G^{Ham}_{\theta,n} \otimes K_{M,M}
\]
(if, respectively, $(x,y)$ is an edge in $G^{Ham}_{\theta,n}$). Here $K_{M,M}$ is a complete bipartite graph with $M=2^m$ vertices in each part, and $\otimes$ denotes the usual tensor product of bipartite graphs.

Keeping in mind  Example~\ref{lines-and-points}, Example~\ref{quasi-distance},  and Example~\ref{hamming}, we obtain that for 
 \emph{most}  edges $(x',y')$ in $\widehat{G^{Pl}_n} $, in $\widehat{G^{Euc}_n} $,  and in $\widehat{G^{Ham}_{\theta,n}} $ we have
\[
\begin{array}{ccccc}
C(x') &\eqp& n &+& m ,\\
C(y') &\eqp& n &+ & m,\\
C(x',y') &\eqp& 1.5n &+& 2m.
\end{array}
\]
 
\section{Bounds with the Spectral Method}

\subsection{Information Inequalities from the Graph Spectrum}\label{ss:spectre}

In this section we show that the spectral properties of a graph can be used to prove information-theoretic inequalities for pairs $(x,y)$ corresponding to the edges in this graph. We start with a reminder of the standard considerations involving the {spectral gap} of a graph. 

Let $G = (L\cup R, E)$ be a regular bipartite  graph of degree $D$ on $2N$ vertices  ($|L|=|R|=N$,  each edge $e\in E$ connects  one vertex from $L$ with another one from $R$, and each vertex is incident to exactly $D$ edges). The adjacency matrix of such a graph is a $(2N) \times (2N)$  zero-one matrix $H$ of the form
\[
\left(
\begin{array}{cc}
0 & J\\
J^\top &0
\end{array}
\right)
\]
(the $N\times N$ submatrix $J$ is usually called \emph{bi-adjacency} matrix of the graph;
$J_{ab} =1$ if and only if there is an edge between the $a$-th vertex in $L$ and the $b$-th vertex in $R$). Let 
\[
\lambda_1\ge\lambda_2\ge\ldots\ge\lambda_{2N}.
\]
be the eigenvalues of $H$. Since $H$ is symmetric, all $\lambda_i$ are real numbers. It is well known that for a bipartite graph
the spectrum is symmetric, i.e., $\lambda_i= - \lambda_{2N-i+1}$ for each $i$. As the degree of each vertex in the graph is equal to $D$, we have $\lambda_1=-\lambda_{2N} =D$. 
We focus on the second eigenvalue of the graph $\lambda_2$; we are interested in graphs such that $\lambda_2\ll \lambda_1$ (that is, the \emph{spectral gap} is large).
\begin{remark}
If the bi-adjacency matrix of the graph is symmetric, then the spectrum of the $(2N)\times (2N)$ matrix $H$ consists of the eigenvalues of the $N\times N$ matrix $J$ and their opposites. This observation makes the computation of the eigenvalues simpler.

It is immediately clear that the bi-adjacency matrices of the bipartite graphs from Example~\ref{quasi-distance} and Example~\ref{hamming} are symmetric. For Example~\ref{lines-and-points} this is also true, since a point with coordinates $(x,y)$ and a line indexed $(a,b)$ are incident if 
$
 a\cdot x - y- b = 0.
$
\end{remark}

In the rest we will use the fact that for the graphs from Example~\ref{lines-and-points} and Example~\ref{quasi-distance} the value of $\lambda_2$ is much less than $\lambda_1=D$:
 \begin{lemma}[see lemma 5.1 in \cite{RVW1}]\label{l:eigenvalue-points-lines}
For the bipartite graph $G^{Pl}_n$ from Example~\ref{lines-and-points} (incident points and lines on plane $\mathbb{F}_q^2$) the second eigenvalue is equal to $\sqrt{q} = \sqrt{D}$.
\end{lemma}
\begin{remark}
We prove the main result of this paper (Theorem~\ref{th:main}) for the construction of $(x,y)$ from  Example~\ref{lines-and-points}.
The same result can be proven for a similar (and even somewhat more   symmetric) construction: we can take lines and points  in the \emph{projective} plane over a finite field.
The construction based on the projective plane has spectral properties similar to Lemma~\ref{l:eigenvalue-points-lines}.
\end{remark}

\begin{lemma}[see theorem~3 in \cite{quasi-dist}]\label{l:eigenvalue-euclid-dist}
For the bipartite graph $G^{Euc}_n$ from  Example~\ref{quasi-distance} (a discrete version of the Euclidean distance) the second eigenvalue is equal to $O(\sqrt{q}) = O(\sqrt{D})$.
\end{lemma}

\begin{remark}\label{r:tensor-product-eigenvalues}
For the tensor product of two graphs $G_1 \otimes G_2$, the eigenvalues can be obtained as pairwise products of the eigenvalues of $G_1$ and $G_2$.  So, for the graph $\widehat{G^{Pl}_n} $ (see p.~\pageref{hat-g-pl}) the eigenvalues are all pairwise products of the graph of incidents lines and points  ${G^{Pl}_n} $ and the complete bipartite graph $K_{M,M}$.
For  ${G^{Pl}_n} $ the first eigenvalue $D$ and the second eigenvalue $\sqrt{D}$. 
The bi-adjacency matrix of $K_{M,M}$ is the $M\times M$ matrix with $1$'s in each cell. 
It is not hard to see that its maximal eigenvalue is $M$ and all other eigenvalues are $0$. Therefore, 
 the first eigenvalue of $\widehat{G^{Pl}_n} $  is equal to $MD$ and the second one is equal to $M\sqrt{D}$. A similar observation is valid for  ${G^{Euc}_n} $.

 \end{remark}

 It is well known that the graphs with a large gap between the first and the second eigenvalues have nice combinatorial properties (vertex expansion, strong connectivity, mixing). 
One version of this property is expressed by  the expander mixing lemma, which was observed by several researchers (see, e.g.,  \cite[lemma 2.5]{HLW1} or  \cite[theorem 9.2.1]{AS1}).
We use a variant of the expander mixing lemma for bipartite graphs (see \cite{bipartite-mixing-lemma}):
\begin{lemma}[Expander Mixing Lemma for bipartite graphs]\label{mixing-lemma}
	Let $G=(L\cup R,E)$ be a regular bipartite graph, where $|L|=|R|=N$ and each vertex has degree $D$. Then for each $A\subseteq L$ and $B\subseteq R$ we have
	\[
	\left|E(A,B)-\frac{D\cdot |A|\cdot|B|}{N} \right|\leq \lambda_2 \sqrt{|A|\cdot|B|},
	\]
	where $\lambda_2$ is the second largest 
	eigenvalue of the adjacency matrix of $G$ and $E(A,B)$ is the number of edges between $A$ and $B$.
\end{lemma}
\begin{remark}
In what follows we apply Lemma~\ref{mixing-lemma} to the graphs with a large gap between $D$ and $\lambda_2$. This technique is pretty common. See, e.g., \cite[Theorem~3]{mixing-lines-points} where the Expander Mixing Lemma was applied to the graph from Example~\ref{lines-and-points}.  
Due to technical reasons, we will need to apply the Expander Mixing Lemma not only to the graph $G^{Pl}_n $  from Example~\ref{lines-and-points} and  $G^{Euc}_n $  from Example~\ref{quasi-distance} 
but also to the tensor product of these graphs and a complete bipartite graph, see below.
\end{remark}
In what follows we use a straightforward corollary of the expander mixing lemma:
\begin{corollary}\label{mixing-lemma-large-sets}
(a) Let $G=(L\cup R,E)$ be a graph satisfying the same conditions as in Lemma~\ref{l:graph-complexity-profile} with $\lambda_2=O(\sqrt D)$.  Then for $A\subseteq L$ and $B\subseteq R$
such that $|A| \cdot |B| \ge N^2/D$ we have 
	\begin{equation}
		\label{eq:mixing-easy}
		\left|E(A,B)\right|  = O\left(  \frac{D\cdot |A|\cdot|B|}{N}  \right).
	\end{equation}
	
(b)  Let $G=(L\cup R,E)$ be the same graph as in~(a),  and let $K_{M,M}$ be a complete bipartite graph for some integer $M$. Define the tensor product of these graphs $\hat G := G\otimes K_{M,M}$ (this is a bipartite graph $(\hat L\cup \hat R, \hat E)$  with $|\hat L |= |\hat R| =  N\cdot M$,   with degree $D\cdot M$). 

Then for all subsets $A \subset \hat L$  and $B\subset \hat R$ such that  $|A| \cdot |B| \ge (MN)^2/D$ inequality~\eqref{eq:mixing-easy} holds true.
\end{corollary}
\begin{proof}
(a) From Lemma~\ref{mixing-lemma} it follows that
	\begin{equation}
		\label{eq:mixing}
		\left|E(A,B)  \right| \le \frac{D\cdot |A|\cdot|B|}{N}  + \lambda_2\sqrt{|A|\cdot|B|}
	\end{equation}
	Assuming that $\lambda_2=O(\sqrt D)$ and $|A|\cdot|B|\ge N^2/D$ we conclude that the first term on the right-hand side of \eqref{eq:mixing} is dominating:
	\[
		\lambda_2\sqrt{|A|\cdot|B|}=O\left(\frac{D\cdot|A|\cdot|B|}N\right)
	\]
	Given this and Lemma~\ref{mixing-lemma} we obtain \eqref{eq:mixing-easy}.

\smallskip

(b)  Let us recall that the eigenvalues of $G\otimes K_{M,M}$  are pairwise products of 
the eigenvalues of $G$ and $K_{M,M}$. Therefore, the maximal eigenvalue of $G\otimes K_{M,M}$ is $MD$ and the second one is $O(M\sqrt{D})$, see Remark~\ref{r:tensor-product-eigenvalues}. The rest of the proof is similar to the case~(a).
\end{proof}

Now we translate the combinatorial property of  \emph{mixing} in the information-theoretic language. We show 
that a large spectral gap in a graph implies some inequality for Kolmogorov complexity that is valid for each pair of adjacent vertex in this graph. We do it in the next lemma, which  is the main technical ingredient of the proof of our main result.
 \begin{lemma}\label{l:graph-spectral-gap}
Let $G=(L\cup R,E)$ be a bipartite graph satisfying the same conditions as in Lemma~\ref{l:graph-complexity-profile}, with  $|L|=|R|=N=2^{n+O(1)}$ and degree $D=O(\sqrt N)$. Assume also that 
the second largest eigenvalue of this graph is $\lambda_2=O(\sqrt{D})$. 
Let $K_{M,M}$ be a complete bipartite graph for some $M=2^m$. Define the tensor product of these graphs $\hat G := G\otimes K_{M,M}$. 

For each edge $(x,y)$ in $\hat G$ and for all $w$ one of the following inequalities is true: either 
\begin{equation}\label{eq:graph-spectral-gap-cond}
 C(x|w)+C(y|w)\le 1.5n+2m + O(\log k)
\end{equation}
or
\begin{equation}\label{eq:graph-spectral-gap}
I(x:y|w)\ge0.5n-O(\log k),
\end{equation}
where $k=n+m$. 
\end{lemma}
\begin{remark}\label{r:spectral-examples}
Note that  Lemma~\ref{l:graph-spectral-gap} applies to the graphs from Example~\ref{lines-and-points} and Example~\ref{quasi-distance} due to Lemma~\ref{l:eigenvalue-points-lines} and Lemma~\ref{l:eigenvalue-euclid-dist} respectively.
\end{remark}
\begin{proof}
Denote $a=C(x|w)$ and $b=C(y|w)$. If $a+b \le 1.5n+2m + O(\log k)$, then \eqref{eq:graph-spectral-gap-cond} is true and we are done. In what follows we assume that $a+b> 1.5n+2m + c \log k$
(for some constant $c$ to be specified later), and our aim is to prove \eqref{eq:graph-spectral-gap}.

 Let $A$  be the set of all $x'\in L$ such that $C(x'|w)\le a$
and  $B$ be the set of all $y'\in R$ such that $C(y'|w)\le b$.  Note that by definition $A$ contains $x$ and $B$ contains $y$.
In what follows we show that for all  pairs $(x',y')\in (A\times B) \cap E$ we have  $C(x',y')\lep a+b-0.5n$.

	\smallskip

	\noindent
	\emph{Claim~1.} We have $\big|\log  |A| - a\big| \le O(\log k)$ and $\big|\log  |B| - b\big| \le O(\log k)$.

	\smallskip
	
	\noindent
	\emph{Proof of the claim~1:} 
	We start with a proof of the upper bounds for the cardinalities of $A$ and $B$. Each element of $A$ can be obtained from $w$ with a programs (description) of length at most $a$. Therefore, the number of elements in $A$ is not greater than the number of such descriptions, which is at most $1+2+\ldots+2^a<2^{a+1}$. Similarly, the number of elements in $B$ is less than $2^{b+1}$.

	Now we proceed with the lower bounds. Given $w$ and an integer number $a$ we can take all programs of size at most $a$, apply them to $w$ and run in parallel. As  some programs converge, we will discover one by one all elements in $A$ (though we do not know when the last stopping program terminates, and when the last element of $A$ is revealed). The element $x$ must appear in this enumeration. Therefore, we can identify it given its position in this list, which requires only $\log|A|$ bits. Thus, we have
	\[
	C(x|w) \le \log|A| + O(\log k)
	\]
	(the logarithmic additive term is needed to specify the binary expansion of $a$).
	On the other hand, we know that $C(x|w)=a$. It follows that 
	$
	\log  |A| \ge {a- O( \log k)},
	$
	and we are done.  The lower bound $\log  |B| \ge {b- O( \log k)}$ can be proven in a similar way.

	\smallskip

	\noindent
	\emph{Claim~2.} The number of edges between $A$ and $B$ is rather small:
	\[
	|(A\times B) \cap E|  \le O\left(\frac{D\cdot |A| \cdot |B|}{N} \right).  
	\]

	\noindent
	\emph{Proof of the claim~2:} By Claim~1 we have $|A| = 2^{a + O(\log k)}$ and $|B| = 2^{b + O(\log k)}$. Since $a+b>1.5n+2m+c\log k$ we obtain
	\[
	|A|\cdot|B|=2^{a+b+c\log k\pm O(\log k)}>2^{1.5n+2m}=(NM)^2/D  
	\]	
	(here we fix the constant $c$: the term $O(\log k)$ can be negative but it should be compensated by the positive term $c\log k$).
	Hence, we can apply Corollary~\ref{mixing-lemma-large-sets}~(b) and obtain the claim.

	\smallskip

	\noindent
	\emph{Claim~3.} For all pairs $(x',y')\in (A\times B) \cap E$ we have 
	\[
	 C(x',y'|w) \le \log|E(A,B)| + O(\log k).
	\]

	\noindent
	\emph{Proof of the claim~3:} Given a string $w$ and the integer numbers $a$, $b$, we can run in parallel all programs of length at most $a$ and $b$ (applied to $w$) and reveal one by one all elements of $A$ and $B$. If we have in addition the integer number $n$, then we can construct the graph $G$ and enumerate all the edges between $A$ and $B$ in the graph $G$. The pair $(x',y')$ must appear in this enumeration. Therefore, we can identify it by its ordinal number in this enumeration. Thus, 
	\[
	C(x',y'|w)\le\log|E(A,B)|+O(\log k),
	\]
	where the logarithmic term involves the binary expansions of $n$, $a$, and $b$.

	\smallskip

	Now we can finish the proof of the lemma.
	By claim~3, we have
	 \[
	 C(x',y'|w)\le\log|E(A,B)|+O(\log k)
	\]
	for all pairs $x',y'\in(A\times B)\cap E$. By using claim~2, we obtain 
	\[
	C(x',y'|w)\le\log D+\log|A|+\log|B|-\log N+O(\log k).
	\]
	With claim~1 this rewrites to
	\begin{equation}\label{eq:xy|w<a+b-0.5n}\nonumber
	C(x',y'|w)\le a+b-0.5n + O(\log k).
	\end{equation}
	Now we apply this inequality to the initial  $x$ and $y$:
	\[
	\begin{array}{rcl}
	I(x:y|w)&\eqp&C(x'|w)+C(y'|w)-C(x',y'|w)\\
	           &\gep& a+b-(a+b-0.5n)+O(\log k)\eqp 0.5n.
	\end{array}
	\]
\end{proof}
\begin{remark}\label{r:lines-points-history}
The property of ``extractability'' of the mutual information for neighboring vertices in the graph from Example~\ref{lines-and-points} 
was studied  in \cite{muchnik1998common,cmrsv}. Instead of the expander mixing lemma,  these papers used the fact that this graph has no cycles of length $4$ 
(since two  lines  cannot intersect at two different points). With this technique, it was proven in \cite{cmrsv}  that for the same $x$ and $y$ as above and for every $w$
\begin{equation}\label{eq:chernov-et-al}
C(w) \le 2C(x|w) +2C(y|w)+O(\log (|x|+|y|+|w|)).
\end{equation}
(Technically, \cite{cmrsv} only considered the case $m=0$; however, the same argument applies to an arbitrary $m$, \cite{zyavgarov}.)
This result is incomparable with Lemma~\ref{mixing-lemma-large-sets}. On the one hand, \eqref{eq:chernov-et-al} is weaker then \eqref{eq:graph-spectral-gap}. 
Indeed, \eqref{eq:graph-spectral-gap} together with  $I(x:y)\eqp0.5n$ imply
 \begin{equation}
 \nonumber
C(w) \le C(x|w) +C(y|w)+O(\log (|x|+|y|+|w|))
\end{equation}
(without factors $2$ in the right-hand side).
On the other hand, \eqref{eq:chernov-et-al} is true for every $w$, while  \eqref{eq:graph-spectral-gap} is proven only under the assumption that~\eqref{eq:graph-spectral-gap-cond} is false.
 
 The bound \eqref{eq:chernov-et-al} is not strong enough to prove our main theorem.
\end{remark}

\subsection{Information Inequalities for a Secret Key Agreement}\label{ss:ineq}

In this section we prove some information-theoretic inequalities that hold true for the objects involved in a communication protocol: the inputs given to Alice and Bob, the transcript of the communication, and the final result computed by Alice and Bob.

The intuitive meaning of the argument below can be explained as follows.
We would like to reduce every transcript of a communication protocol to a ``natural'' form, where the complexity
profile for the triple
\[
\langle \text{Alice' data}, \text{Bob's data}, \text{transcript}\rangle
\]
looks as it is shown in the diagram in Fig.~\ref{x-y-t-natural}.  The key property of this profile is that the triple mutual information
\[
\begin{array}{l}
I( \text{Alice' data} :  \text{Bob's data} : \text{transcript} ) =\\
I( \text{Alice' data} :  \text{Bob's data} ) - I( \text{Alice' data} :  \text{Bob's data}\ |\ \text{transcript} ) 
\end{array}
\]
is equal to $0$ (as usual, within a logarithmic precision). In Shannon's entropy setting, a similar property would mean that
the \emph{external information complexity} of the protocol is equal to its \emph{internal information complexity}.
In other words, the quantity  of information learned from the transcript by the eavesdropper 
is equal to  the sum of information quantities  learned from the transcript by Alice and Bob.

In Fig.~\ref{x-y-t-natural}, the transcript is split into two components, $t_A$ and $t_B$. These components are independent with each other (the mutual information $I(t_A:t_B)$ is negligibly small), one of them is independent of Alice' data and the other one is independent of Bob's data. Intuitively, one of these component should consist of all messages of Alice, and the other one should consist of all messages of Bob.
			\begin{figure}[h]
				\centering
				\begin{tikzpicture}[scale=0.70]
				  \draw \firstcircle node[above left] {$\ge 0$};
				  \draw \secondcircle node [above right] {$\ge 0$};
				  \draw \thirdcircle node [below] {$0$};
				  \node {$ 0$};
				  \node at (90:1.55) {$ \ge 0 $};
				  \node at (210:1.75) {$t_A$};
				  \node at (330:1.75) {$t_B$};
				  \node at (150:5.75) {\textsl{Alice' data}};
				  \node at (30:5.75) {\textsl{Bob's data}};
				  \node at (270:4.75) {\textsl{transcript} = $\langle t_A,t_B\rangle$};
				\end{tikzpicture}
				\caption{Complexity profile for Alice' and Bob's data and a ``natural''  communication transcript.}
				\label{x-y-t-natural}
			\end{figure}
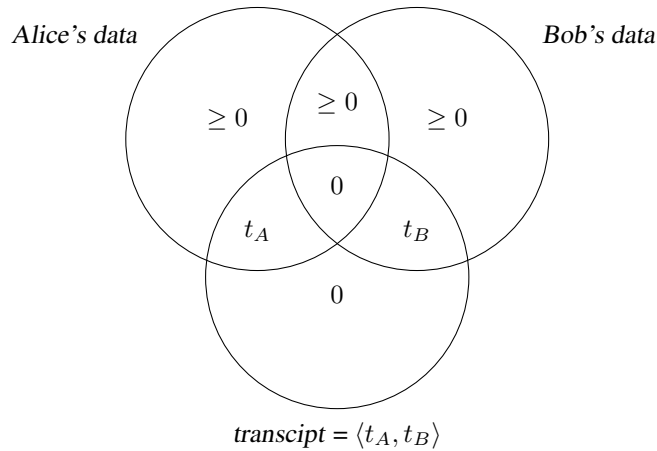

The diagram in Fig.~\ref{x-y-t-natural} represents well enough the situation for ``natural'' examples of communication protocols. Indeed,
in naturally constructed  examples, each message of Alice can be chosen in such a way that it has virtually 
no mutual information with Bob's input (even given all the previous messages of the protocol as a condition). 
Similarly, each message of Bob can be chosen in such a way that it has virtually  no mutual information with Alice' input 
(again, given all the previous messages of the protocol as a condition). 

In general case, the transcript of the protocol cannot be represented as shown in Fig.~\ref{x-y-t-natural}. Indeed, Alice may put in her messages some information that is already known to Bob, and vice-versa. However, it seems plausible that the protocols which are \emph{not} natural (in the sense explained above) are also not optimal in the sense of communication complexity. Indeed, if Alice and Bob send to each other excessive information, then it looks believable that we can ``compress'' excessively long messages of Alice and Bob and, therefore, make the communication complexity smaller. 
The aim of this section is to formulate and prove a more precise version of this intuitive guess.

We will prove that the transcript $t$ of \emph{every} communication protocol can be reduced in some sense to the form shown in Fig.~\ref{x-y-t-natural}. More precisely, we will extract from a generic transcript $t$ a digital fingerprint $t'$ so that the complexity profile of the triple $\langle  \text{Alice' data}, \text{Bob's data},t'\rangle$ looks as shown in Fig.~\ref{x-y-t-natural}. Moreover, this ``extraction''  preserves all essential information of the original transcript: Alice and Bob can retrieve  from the new $t'$ (almost) the same information as they could obtain from the original transcript $t$. 
The formal version of this statement is given in Lemma~\ref{l:simplifying-t}(e), see below.

The reduction of an arbitrary communication transcript $t$ to its compressed version $t'$ is quite non-trivial, especially for protocols with an unbounded number of rounds. 
We use in the proof two important technical tools: 
Muchnik's technique of conditional descriptions (see Proposition~\ref{p:muchnik}) and the lemma on the non-negativity of the triple mutual information for communication transcripts (see Lemma~\ref{l:triple-info}).

\smallskip 

The following lemma was proven in \cite{RZ1} (see also a similar result proven for Shannon entropy in \cite{kaced}):
\begin{lemma}[\cite{RZ1}]\label{l:triple-info}
Let us consider a deterministic communication protocol  with two parties. Denote by $x$ and $y$  the inputs of the parties, and denote  by  $t=t(x,y)$ the transcript of the communication between the parties. Then 
$
I(x:y:t) \gep 0,
$
see Fig.~\ref{x-y-t}.
\end{lemma}
			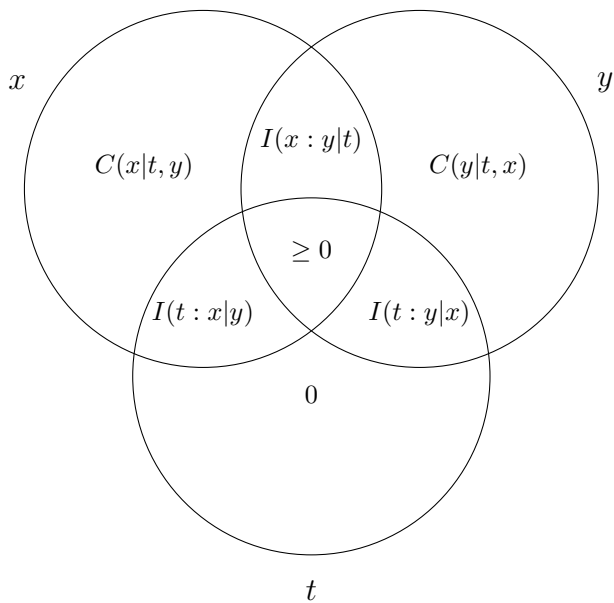
\begin{figure}[h]
				\centering
				\begin{tikzpicture}[scale=0.95]
				  \draw \firstcircle node[above left] {$C(x|t,y)$};
				  \draw \secondcircle node [above right] {$C(y|t,x)$};
				  \draw \thirdcircle node [below] {$0$};
				  \node {$\geq 0$};
				  \node at (90:1.55) {$ I(x:y|t)$};
				  \node at (210:1.75) {$I(t:x|y)$};
				  \node at (330:1.75) {$I(t:y|x)$};
				  \node at (150:4.75) {\Large $x$};
				  \node at (30:4.75) {\Large $y$};
				  \node at (270:4.75) {\Large $t$};
				\end{tikzpicture}
				\caption{Complexity profile for  inputs $x,y,$ and the transcript $t$ of a communication protocol with given inputs. Note that $C(t|x,y)$ is negligibly small (we can compute $t$ by simulating the communication protocol) and  $I(x:y|t)\lep I(x:y)$ due to Lemma~\ref{l:triple-info}.}\label{x-y-t}
			\end{figure}
\begin{proposition}[Muchnik's theorem on conditional descriptions, \cite{muchnik}]\label{p:muchnik}
(a) Let $a$ and $b$ be arbitrary strings of length at most $n$. Then there exists a string p of length $C(a|b)$ such that
\begin{itemize}
\item	$C(p|a) = O(\log n)$, 
\item $C(a|p,b)=O(\log n)$.
\end{itemize}
(b) Let $a, b_1, b_2$ be arbitrary strings of length at most $n$. Then there exist  strings $q_1$, $q_2$  of length $C(a|b_1)$ and   $C(a|b_2)$ respectively such that
\begin{itemize}
\item	$C(q_j|a) = O(\log n)$,
\item $C(a|b_j, q_j) =O(\log n)$ 
\end{itemize}
for $j = 1,2$; we may also require  that one of the strings $q_1, q_2$ is a prefix of another one.
As usual, the constants in $O(\log n)$-notation do not depend on $n$.
\end{proposition}

\begin{remark}
In Proposition~\ref{p:muchnik}(a) the string $p$ can be interpreted as an almost (up to $O(\log n)$) shortest description of $a$ conditional on $b$ that satisfies a nice additional  property: it can be easily computed given $a$. In other words, $p$ is a ``digital fingerprint'' of $a$ (it is easy to obtain $p$ from $a$) such that we can reconstruct $a$ given this fingerprint and~$b$.

Similarly, in Proposition~\ref{p:muchnik}(b) the strings $q_1$ and $q_2$ can be interpreted as almost shortest descriptions of $a$ given $b_1$ and $b_2$ respectively.
The non-trivial part of (b) is the requirement that one of the strings $q_1,q_2$ (the shorter one) is a prefix of the other (longer) one. In particular, if $C(a|b_1)=C(a|b_2)$, then $q_1=q_2$, and we can use \emph{one and the same} shortest program to transform $b_1$ or $b_2$ into $a$.

Since $p$  in Proposition~\ref{p:muchnik}(a) is an almost optimal description of $a$ given $b$, we can conclude that $p$  is almost incompressible given $b$, i.e., $C(p|b) = |p|+O(\log n)$. Similar properties hold true for $q_1$and $q_2$ in Proposition~\ref{p:muchnik}(b): these strings are almost incompressible conditional on $b_1$ and $b_2$ respectively. We will use this observation in the proof of the next lemma.

This proposition can be proven by an ad~hoc argument, \cite{muchnik}, or by using the generic technique of extractor, as shown in \cite{buhrman2001resource}. For a discussion of different proofs of this theorem we refer the reader to \cite{musatov2011variations}.
\end{remark}
Now we can prove the main technical result of this section.
\begin{lemma}\label{l:simplifying-t}
Assume a deterministic communication protocol for  two parties on inputs $x$ and $y$ gives transcript $t$ and denote $n=C(x,y,t)$.

\smallskip

\noindent
(a) $C(t|x,y) = O(\log n)$.

\smallskip

\noindent
(b) $C(t|x) \eqp I(t:y|x) $.

\smallskip

\noindent
(c) $C(t|y) \eqp I(t:x|y) $.

\smallskip

\noindent
(d) $ C(t|x) + C(t|y) \eqp  I(t:x|y) + I(t:y|x)+ O(\log n) \lep C(t)  $.

\smallskip

\noindent
(e) There exist  $t_x$ and $t_y$ such that 
\begin{itemize}
\item $C(t_x) = C(t|x)$  and $C(t_y) = C(t|y)$, 
\item $C(t_x|t) = O(\log n)$ and $C(t_y|t) = O(\log n)$,
\item $C(t| t_x,x) = O(\log n)$ and  $C(t| t_y,y) = O(\log n)$,
\item $C(t_x,t_y) \eqp C(t_x) + C(t_y) $.
\end{itemize}
Speaking informally, $t_x$ and $t_y$ are ``fingerprints'' of $t$ that can play the roles of (almost) shortest descriptions of $t$
conditional on $x$ and $y$ respectively. The last condition means that the mutual information between $t_x$ and $t_y$ is negligibly small. 

The complexity profile for $x$, $y$, and $\langle t_x, t_y\rangle$ is shown in Fig.~\ref{x-y-t'}.

\end{lemma}
			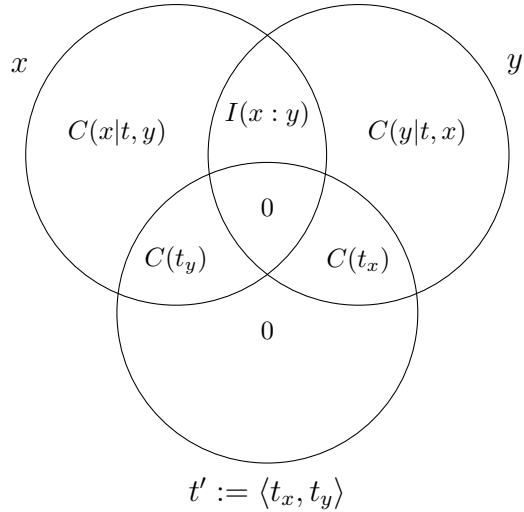
\begin{figure}[h]
				\centering
				\begin{tikzpicture}[scale=0.80]
				  \draw \firstcircle node[above left] {$C(x|t,y)$};
				  \draw \secondcircle node [above right] {$C(y|t,x)$};
				  \draw \thirdcircle node [below] {$0$};
				  \node {$ 0$};
				  \node at (90:1.55) {$ I(x:y)$};
				  \node at (210:1.75) {$C(t_y)$};
				  \node at (330:1.75) {$C(t_x)$};
				  \node at (150:4.75) {\Large $x$};
				  \node at (30:4.75) {\Large $y$};
				  \node at (270:4.75) {\Large $t':=\langle t_x,t_y\rangle$};
				\end{tikzpicture}
				\caption{Complexity profile for  $x,y,$ and $t':=\langle t_x,t_y\rangle$ from Lemma~\ref{l:simplifying-t}. Note that $C(t_x)= I(x:t|y)$, $C(t_y)= I(y:t|x)$, and $I(x:y|t')= I(x:y).$}\label{x-y-t'}
			\end{figure}

\begin{remark}
Lemma~\ref{l:simplifying-t} 
is a technical statement, and its claim~(e)  might look artificial. However, this claim has  an intuitive motivation, as explained in the beginning of this section. The ``compressed'' components $t_x$ and $t_y$ correspond to the components $t_B$ and $t_A$ of a ``natural'' transcripts shown in Fig.~\ref{x-y-t-natural}.

In an arbitrary protocol the messages of Alice and Bob may have positive mutual information with the data of their counterparts: Alice may send a message (partially) known to Bob, and Bob may send a message (partially) known to Alice,
in absolute terms or conditional on the previously sent messages.
We can find  ``compressed'' descriptions of each individual message in the protocol using Muchnik's method, Proposition~\ref{p:muchnik}. Indeed, for each of Alice' message there exists a compressed description that looks incompressible from Bob's point of view, conditional on the data available to Bob and the previous messages of Alice; similarly, for each of Bob's messages there exists a compressed description that looks incompressible from Alice' perspective.  
By combining together the obtained compressed codes of separate messages, we can try to construct $t_x$ and $t_y$ required in the claim~(e) of the lemma.
However, this approach fails for protocols  with unbounded number of rounds. Indeed, ``compressing'' each individual message with Proposition~\ref{p:muchnik}(a) costs us a logarithmic error term;  as the number of messages is large, the sum of logarithmic terms grows beyond control.
This is why  we have to use a less intuitive argument based on  Lemma~\ref{l:triple-info}, which helps 
to handle the transcript of a protocol in one piece, without splitting it into separate messages, as we do in the proof below.
\end{remark}

\begin{proof}
(a) follows trivially from the fact that $t$  can be computed given $(x,y)$ (we may simulate the communication protocol on the given inputs). Note that the constant in the term $O(\cdot)$ includes implicitly a description of the communication protocol
(we assume that the protocol has a description of size $O(1)$, see the discussion on p.~\pageref{par:uniform-protocol}).

\smallskip

(b) For all $x,y,t$ we have
\[
C(t|x) \eqp C(t|x,y) + I(t:y|x) .
\]
The term $C(t|x,y)$ vanishes due to~(a), and we are done.

\smallskip

(c) Is similar to (b).

\smallskip

(d) A routine check shows that for all $x,y,t$ we have
\[
C(t) \eqp I(t:x|y) + I(t:y|x) +( I(x:y) - I(x:y|t) ) + C(t|x,y) .
\]
Due to Lemma~\ref{l:triple-info} we have $I(x:y) - I(x:y|t) \eqp I(x:y:t)\gep0$,
so (d) follows.

\smallskip

(e) First, we apply Proposition~\ref{p:muchnik}(a) with $a=t$ and $b=x$;  we obtain a string $p$ of length $C(t|x)$ such that 
\begin{itemize}
\item $C(p|t)= O(\log n)$ and
\item $C(t|p,x)= O(\log n)$.
\end{itemize}
From (b) we have $C(p)\eqp I(t:y|x)$. So we can let $t_x:=p$.

Observe that 
\[
C(t|t_x) \eqp  C(t) - C(t_x) \eqp C(t) - I(t:y|x)   \ge C(t|y) 
\]
(this inequality follow from (d)).
Now we apply Proposition~\ref{p:muchnik}(b) with $a=t$, $b_1=y$, and $b_2=t_x$. We obtain $q_1,q_2$ 
of length  $C(t|y)$ and $C(t|t_x)$ respectively  such that
\begin{itemize}
\item $q_1$ and $q_2$ are comparable (one of these strings is a prefix of the other one),
\item $C(q_1|t) = O(\log n)$ and $C(q_2|t) = O(\log n)$,
\item $C(t|y, q_1) = O(\log n)$, where $q_1$ is the prefix of $q$ having length $C(t|y)$, and
\item $C(t|t_x, q_2) = O(\log n)$, where $q_2$ is the prefix of $q$ having length $C(t|t_x)$.
\end{itemize}
Note that the length of $q_2$ is not less (up to $O(\log n)$) than the length of $q_1$. Since  $q_2$ is incompressible conditional on $t_x$, the shorter prefix $q_1$ must be also incompressible conditional on $t_x$. Thus, $t_x$ and $q_1$ are independent. We let $t_y:=q_1$, and (e) is proven. 
\end{proof}

\subsection{Proof of Theorem~\ref{th:main}}

Now we are ready to combine the spectral technique from Section~\ref{ss:spectre} and the information-theoretic technique from Section~\ref{ss:ineq} and prove our main result.
\begin{proof}[Proof of Theorem~\ref{th:main}]
Let us take a pair of $(x,y)$ from Example~\ref{lines-and-points} or Example~\ref{quasi-distance}. We know that it satisfies \eqref{eq:graph-complexity-profile} and, therefore, \eqref{eq:complexity-profile}.  Assume that  in a communication protocol $\pi$  Alice and Bob (given as inputs $x$ and $y$ respectively) agree on a secret key $z$ of size $\delta(n)$. We will prove a lower bound on the communication in this protocol. To simplify the notation, in what follows we ignore the  description of $\pi$  in all complexity terms (assuming that it is a constant, which is negligible compared with $n$).
 
In this proof we will deal with four objects: the inputs $x'=\langle x, r_A\rangle$ and  $y'=\langle y, r_B\rangle$, the transcript $t$, and the output of the protocol (secret key) $z$. Our aim is to prove that $C(t)$ cannot be much less than $0.5n$. This is enough to  conclude that the length of the transcript measured in bits (which is exactly the communication complexity of the protocol)  also  cannot be much less than $0.5n$. Due to some technical reasons that will be clarified below we need to reduce in some sense the sizes of $t$ and $z$. 

\smallskip

\emph{Reduction of the key}.
First of all, we reduce the size of $z$. This step might seem counter-intuitive: we make the assumption of the theorem \emph{weaker} by suggesting that Alice and Bob agree on a rather small secret key. We know from \cite{RZ1} that $C(z)$ can be pretty large (more specifically, it can be of complexity $0.5n + O(\log n)$). However, we prefer to deal with protocols where Alice and Bob agree on a moderately small (but still not \emph{negligibly} small) key. To this end we may need to degrade the given communication protocol and reduce the size of the secret key to the value $\mu\log n$ (the constant $\mu$ to be chosen later). It is simple to make the protocol weaker: if the original protocol provides a common secret key $z$ of larger size, then in the degraded protocol Alice and Bob can take only the $\delta(n)$ first bits of this key. Thus, without loss of generality, we may assume that the protocol gives a secret key $z$ with complexity $\delta(n)=\mu \log n$.

\smallskip

\emph{Reduction of the transcript}.
Now we perform a reduction of $t$. We know from Lemma~\ref{l:triple-info} that $I(x':y':t)$
is non-negative. We want to reduce $t$ to a $t'$ such that $I(x':y':t')$
 is exactly $0$ (here \emph{exactly} means, as usual, an equality that holds up to $O(\log n)$).
To this end, we apply Lemma~\ref{l:simplifying-t} to the triple $(x',y',t)$ and obtain $t_x$ and $t_y$, which play the roles of optimal descriptions of $t$ given the conditions $x'$ and $y'$ respectively. We let $t':=\langle t_x,t_y\rangle$.
Though technically this $t'$ is not a transcript of any communication protocol, we will see that in some sense it behaves similarly to the initial transcript.

We know from Lemma~\ref{l:simplifying-t}(d,e) that $C(t') \lep C(t) $. Thus, to prove the theorem, it is enough to show that 
$
C(t')  \gep 0.5n-2\delta(n) .
$
\begin{lemma} \label{l:t-prime}
For $t'=\langle t_x,t_y\rangle$ we have the following equalities:

\smallskip

(a)
$C(x'|t',z)  \eqp  n   +m -   C(t_y) - \delta(n)$,

\smallskip

(b)
$C(y'|t',z)  \eqp n   +m -   C(t_x) - \delta(n) $,

\smallskip

and
\begin{equation}\label{eq:cond-info-small}
\text{\it (c)}\ I(x':y'|t',z) \eqp  I(x':y') - C(z)+ O(\log n) = 0.5n - \delta(n), \rule{8mm}{0mm}
\end{equation}
 see Fig.~\ref{x'-y'-t'-z'}.
\end{lemma}
			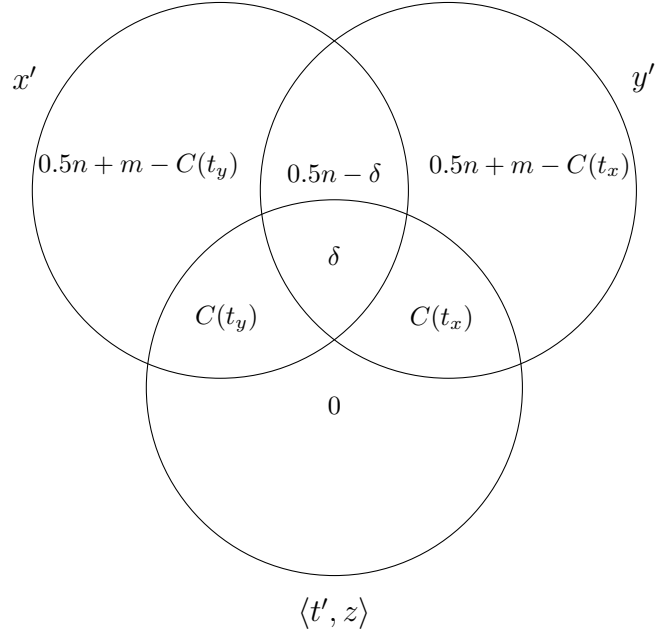
\begin{figure}[h]
				\centering
				\begin{tikzpicture}[scale=0.85]
				  \draw \firstcircle node[above left] {}; 
				    \node at (-2.6,1.2) {$0.5n+m-C(t_y)$};   
				  \draw \secondcircle node [above right] {}; 
				    \node at (2.6,1.2) {$0.5n+m-C(t_x)$};   

				  \draw \thirdcircle node [below] {$0$};
				  \node {$ \delta$};
				  \node at (90:1.10) {$ 0.5n - \delta $};
				  \node at (210:1.65) {$C(t_y)$};
				  \node at (330:1.65) {$C(t_x)$};
				  \node at (150:4.75) {\Large $x'$};
				  \node at (30:4.75) {\Large $y'$};
				  \node at (270:4.75) {\Large $\langle t', z\rangle$};
				\end{tikzpicture}
				\caption{Complexity profile for  the triple that consists of $x',y'$, and $\langle t',z\rangle$ from Lemma~\ref{l:t-prime}.}\label{x'-y'-t'-z'}
			\end{figure}

\begin{proof}[Proof of the lemma]  
The proof is a routine check where we use repeatedly the Kolmogorov--Levin theorem. For~(a) we have
\[
\begin{array}{rcl}
C(x'|t',z)  & \eqp& C(x',t', z)  - C(t',z)   \\
 &&\text{/ from the Kolmogorov--Levin theorem /} \\
& \eqp & C(x') + C(t'|x') + C( z| t',x') - (C(t') + C(z) )   \\
 &&\text{/  since $z$ is independent of $t'$ /} \\
& \eqp & C(x') + C(t'|x') - C(t')  - C(z)   \\
 &&\text{/  $z$ is computable given $t$ and $x'$, so $C(z|x',t')=O(\log n)$  /} \\
& \eqp & C(x') + I(y':t|x' )  -   I(x':t|y')  - I(y':t|x')   -  C(z)  \\
 &&\text{/ from Lemma~\ref{l:simplifying-t}  /} \\
& \eqp & C(x')  -  I(x':t|y')  -  \delta(n)   \\
&\eqp& n   + m  -   C(t_y) - \delta(n) 
\end{array}
\]
The proof of~(b) is similar.

Since $t$ and $z$ can be computed from $(x',y')$ by a simulation of the protocol
and $t'$ has negligibly small complexity conditional on $t$,  we obtain from the Kolmogorov--Levin theorem
\[
\begin{array}{rcl}
C(x' , y'|t',z) & \eqp& C(x', y',t',z) - C(t',z)\\
&&\ \text{ / from the Kolmogorov--Levin theorem /}\\
& \eqp& C(x', y') - C(t',z) \\
&& \text{ / since $t'$ and $z$ have logarithmic complexity conditional on $(x',y')$ /}\\
& \eqp& C(x' , y') - C(t') - C(z)  \\
&& \text{ / since $z$ is incompressible given $t'$ /}\\
& \eqp& 1.5n +2m   -  I(x':t|y')   -  I(y':t|x')  -   \delta(n).
\end{array}
\]
Combining this with (a) and (b) we obtain (c).
\end{proof}

Now we are ready to prove the theorem.  Assume that 
\begin{equation}\label{eq:cc-is-small}
 C(t_x)  +  C(t_y) < 0.5n -2\delta(n) - \lambda \log n.
 \end{equation}
 If the constant $\lambda$ is large enough, then we obtain from Lemma~\ref{l:t-prime}(a,b) 
\begin{equation}\label{eq:graph-spectral-gap-cond-not}
C(x'|t',z) + C(y'|t',z)  > 1.5n + 2m +c\log n.
\end{equation}
By choosing $\lambda$ we can make the constant $c$ in this inequality sufficiently large, so that \eqref{eq:graph-spectral-gap-cond-not} contradicts \eqref{eq:graph-spectral-gap-cond}.
Now we can apply  Lemma~\ref{l:graph-spectral-gap} (the spectral bound applies to Example~\ref{lines-and-points} and Example~\ref{quasi-distance}, see Remark~\ref{r:spectral-examples}), and
\eqref{eq:graph-spectral-gap} rewrites to
\begin{equation}\label{eq:cond-info-large}
I(x':y'|t',z) \ge   0.5n - O(\log n).
\end{equation}
Comparing \eqref{eq:cond-info-small} and \eqref{eq:cond-info-large} we conclude that 
$\delta(n)=O(\log n)$ (the constant hidden in $O(\cdot)$ depends only on the choice of the optimal description method in the definition of Kolmogorov complexity).
This contradicts the assumption $\delta(n)=\mu\log n$, if $\mu$ is chosen large enough. 
Therefore, the assumption in \eqref{eq:cc-is-small} is false
(without this assumption we cannot apply Lemma~\ref{l:graph-spectral-gap} and conclude with \eqref{eq:cond-info-large}).
 
The negation of \eqref{eq:cc-is-small} gives
\[
C(t) \ge C(t_x) + C(t_y) - O(\log n) \ge 0.5n -2 \delta(n) - O(\log n),
\]
and we are done.
\end{proof}

\section{Pairs with a Fixed Hamming Distance}

Theorem~\ref{th:main} estimates communication complexity of the protocol in the worst case. For some classes of inputs $(x,y)$ there might exist more efficient communication protocol. In this section we study one such special class --- the pairs $(x,y)$ from Example~\ref{hamming}. The spectral argument from the previous section does not apply to this example. The spectral gap for the graph from Example~\ref{hamming} is too small: for this graph we have $\lambda_2=\Theta(\lambda_1)$, while in Example~\ref{lines-and-points} and Example~\ref{quasi-distance} we had $\lambda_2=O(\sqrt\lambda_1)$. In fact, the spectrum of the graph from Example~\ref{hamming} can be computed explicitly:
the eigenvalues of this graph are the numbers 
\[
	K_{\theta n}(i)=\sum_{h=0}^{\theta n}(-1)^h\binom ih\binom{n-i}{\theta n-h}\text{ for }i\in\{0,1,\dots n\}
\]
with different multiplicities, see \cite{BCIMG1} and the survey \cite{LZh1}.
In particular, the maximal eigenvalue of this graph is $K_{\theta n}(0) = \binom n{\theta n}$ and its second eigenvalue is $K_{\theta n}(1)=\binom{n-1}{\theta n}-\binom{n-1}{\theta n-1}$. It is not difficult to verify that $K_{\theta n}(1)=\Omega\left(\binom n{\theta n}\right)$ (for a fixed $\theta$ and $n$ going to infinity), so the difference between the first and the second eigenvalues is only a constant factor. Thus, we cannot apply Lemma~\ref{l:graph-spectral-gap} to this graph.

It is no accident that our  \emph{proof} of Theorem~\ref{th:main}  fails on  Example~\ref{hamming}. Actually, the \emph{statement} of the theorem is not true for $(x,y)$ from this example. In what follows we show  that given these $x$ and $y$ Alice and Bob can agree on a secret key of any size $m$ (intermediate between $\log n$ and $n/2$) with communication complexity $\Theta(m )$. 
The positive part of this statement (the existence of a communication protocol with communication complexity $O(m)$) is proven in Theorem~\ref{th:hamming1}. 
The negative part of the statement (the lower bound $\Omega(m)$ for all communication protocols) is proven in Theorem~\ref{th:hamming2}. 

\begin{theorem}\label{th:hamming1}
For every $\delta\in(0,1)$  there exists a two-parties randomized communication protocol $\pi$ such that  given inputs  $x$ and $y$
from Example~\ref{hamming} (a pair of $n$-bit strings with the Hamming distance $\theta n$ and complexity profile \eqref{eq:complexity-profile}) Alice and Bob with probability $>0.99$ agree on a secret key $z$ of size $\delta n/2 - o(n)$ with communication complexity $O(\delta n)$.
(The constant hidden in the $O(\cdot)$ does not depend on $n$ or $\delta$.)
\end{theorem}
\begin{theorem}\label{th:hamming2}
For every $\delta\in(0,1)$ 
for every randomized communication protocol $\pi$ such that for inputs  $x$ and $y$
from  Example~\ref{hamming} Alice and Bob with probability $>0.99$ agree on a secret key $z$ of size $\ge \delta n $,
the communication complexity is at least $\Omega(\delta n)$.
(The constant hidden in the $\Omega(\cdot)$ does not depend on $n$ or $\delta$.)
\end{theorem}

\begin{proof}[Proof of Theorem~\ref{th:hamming1}.]
We start  the proof with a  lemma.
\begin{lemma}\label{l:prefix}
Let $(x,y)$ be a pair from from Example~\ref{hamming} (two $n$-bit strings with the Hamming distance $\theta n$ and complexity profile \eqref{eq:complexity-profile}). Let $m= \delta n$ for some $\delta\in(0,1)$. Denote by $\hat x$ and $\hat y$ the $m$-bit prefixes of $x$ and $y$ respectively. Then 
\begin{itemize}
\item $C(\hat x) = m + o(n)$,
\item $C(\hat y) = m + o(n)$,
\item $I(\hat x: \hat y) = 0.5m + o(n)$,
\item the Hamming distance between $\hat x$ and $\hat y$ is $\theta m + o(n)$.
\end{itemize}
\end{lemma}
\begin{proof}[Proof of lemma.]
Denote by $\hat x'$ and $\hat y'$ the suffixes of length $n-m$ of $x$ and $y$ respectively (so $x$ is a concatenation of $\hat x$ and $\hat x'$, and $y$ is a concatenation of $\hat y$ and $\hat y'$). The idea of the proof is simple: the law of large number guarantees that for the vast majority of pairs $(x,y)$ such that  $\hdist( x,  y)=\theta n$, the fraction of positions where $\hat x$ differs from $\hat y$
and the fraction of positions where $\hat x'$ differs from $\hat y'$ are both close to $\theta$; the pairs violating this rule are uncommon; therefore, Kolmogorov complexity of these ``exceptional'' pairs is small, and they cannot satisfy  \eqref{eq:complexity-profile}. To convert this idea into a formal proof, we need the following technical claim:

\smallskip
\noindent
\emph{Claim:} If the pair $(x,y)$ satisfies \eqref{eq:graph-complexity-profile} and the Hamming distance between $x$ and $y$ is $\theta n$, then
\[
\begin{array}{lcl}
\hdist(\hat x, \hat y)&=& \theta m + o(n), \\
\hdist(\hat x', \hat y') &=& \theta (n-m) + o(n).
\end{array}
\]
\begin{proof}[Proof of the claim.]

Denote
\[
\begin{array}{l}
\theta_1 := \frac{1}{m}\cdot \hdist(\hat x, \hat y),\\
\theta_2 := \frac{1}{n-m}\cdot \hdist(\hat x', \hat y')
\end{array}
\]
(note that  $\theta_1 m + \theta_2 (n-m) = \theta n$; as $\theta$ is fixed, there is a linear correspondence between $\theta_1$ and $\theta_2$).
For a fixed $x$ of length $n$, the number  strings $y$ of the same length that matches the parameters $m,\theta_1,\theta_2$
(i.e., that differ from $x$   
in exactly  $\theta_1m$ bits in the first $m$ positions and in $\theta_2(n-m)$ bits in the last $n-m$ positions)
is 
\[
\begin{array}{rcl}
{m \choose {\theta_1} m} \cdot {n-m \choose {\theta_2 (n-m)}} &=& 2^{h(\theta_1) m + O(\log n)} \cdot 2^{h(\theta_2) (n-m) + O(\log n)} \\
&=&2^{(\frac{m}{n} h(\theta_1) + \frac{n-m}{n} h(\theta_2)) n + O(\log n) } \\
&=&2^{(\delta  h(\theta_1) + (1-\delta)  h(\theta_2)) n + O(\log n) } \\
&\le &2^{ h(\theta) n + O(\log n) } =  2^{ 0.5 n + O(\log n) }, 
\end{array}
\]
where $h(\tau)=-\tau\log\tau - (1-\tau)\log(1-\tau)$ is the binary entropy function.
The last inequality follows from the fact that the  function $h(\tau)$ is concave, and therefore
\begin{equation}\label{eq:concavity}
 \delta h(\theta_1) + (1-\delta) h(\theta_2)  \le  h\left( \delta \theta_1 + (1-\delta) \theta_2  \right) = h(\theta).
\end{equation}
If $\theta_1$ and $\theta_2$ are not close enough to the average value $\theta$, then the gap between the left-hand side and the right-hand side in \eqref{eq:concavity} is getting large. 
More specifically, it is not hard to verify that the difference
\begin{equation}\label{eq:concavity-diff}
 h(\theta) - \delta h(\theta_1) - (1-\delta) h(\theta_2) 
\end{equation}
grows essentially proportionally to the square of $|\theta_1-\theta|$ (as the second term of the Taylor series of the function around the extremum point). So, if $|\theta_1-\theta|$ and $|\theta_2-\theta|$ are getting much larger than $\sqrt{(\log n)/n}$, then the gap \eqref{eq:concavity-diff} becomes much bigger than $(\log n)/n$, and then we obtain
\[
{m \choose {\theta_1} m} \cdot {n-m \choose {\theta_2 (n-m)}}  < 2^{0.5n - \omega(\log n)}.
\]
On the other hand,
\[
C(x,y) \eqp C(x) + C(y|x) \lep n +  \log\left({m \choose {\theta_1} m} \cdot {n-m \choose {\theta_2 (n-m)}}\right).
\]
Thus, the assumption $C(x,y)\eqp1.5n$ can be true only if $\theta_1m=\theta m+o(n)$ and $\theta_2(n-m)=\theta (n-m)+o(n)$.
\end{proof}
Note that
\[
\begin{array}{ccccl}
C(\hat x) &\le& |\hat x| + O(1)  &=& m + O(1),\\
C(\hat y) &\le& |\hat y| + O(1)  &=& m + O(1).
\end{array}
\]
Further, from the Claim it follows that
\begin{equation}\label{eq:cond-compl}
C(\hat y | \hat x) \lep \log{m \choose {\theta m + o(n)} }  = h(\theta) \cdot m + o(n) = 0.5m + o(n).
\end{equation}
Therefore,
\[
C(\hat x, \hat y) \eqp C(\hat x) + C(\hat y | \hat x)  \lep 1.5m + o(n).
\]
Thus, we have proven that
\[
\begin{array}{ccccl}
C(\hat x) &\le& m + o(n),\\
C(\hat y) &\le& m + o(n),\\
C(\hat x,\hat y) &\le& 1.5m + o(n).\\
\end{array}
\]
It remains to show that these three bounds are tight.  To this end, we observe that a similar argument gives the upper bound
\[
C(\hat x',\hat y') \le 1.5m + o(n).
\]
Since
\[
C(\hat x,\hat x',\hat y, \hat y') \eqp C(x,y) = 1.5n + O(1),
\]
we obtain 
\[
C(\hat x,\hat y) \ge 1.5m - o(n).
\]
Due to \eqref{eq:cond-compl}, this implies 
$
C(\hat x) \ge m - o(n), \text{ and similarly }C(\hat y) \ge m - o(n).
$
\end{proof}

Thus, Alice and Bob can take the prefixes of their inputs $x,y$ of size $m=\delta n$. Lemma~\ref{l:prefix} guarantees that these prefixes $\hat x $ and $\hat y$ have the  complexity profile (Kolmogorov complexities and  mutual information)  similar to the complexity profile of the original pair $(x,y)$ scaled with the factor of $\delta$ (up to an $o(n)$-term).
Thus,  Alice and Bob can apply to $\hat x $ and $\hat y$ the communication protocol from Theorem~\ref{th:rz} and end up with a secret key $z$ of size $\delta n/2-o(n)$. It is shown in \cite{RZ1} that communication complexity of this protocol is $C(\hat x|\hat y) + O(\log m)$ (note that it is enough for Alice and Bob to know the complexity profile of $(\hat x, \hat y)$  within a precision $o(n)$, see Remark~5 in \cite{RZ1}).  In our setting this communication complexity is equal to $\delta n/2 + o(n)$.
\end{proof}

\begin{proof}[Proof of Theorem~\ref{th:hamming2}.]
In the proof of the theorem we use two lemmas. The first lemma gives us a pair of simple  information inequalities:
\begin{lemma}[see, e.g., Ineq~6 in \cite{hrsv} or lemma~7 in \cite{mmrv}]\label{l:ineq}
For all binary strings $a,b,c,d$
\[
\begin{array}{cccl}
(i) & C(c) &\lep& C(c|a) + C(c|b) + I(a:b),\\
(ii) & C(c|d) &\lep& C(c|a) + C(c|b) + I(a:b|d).
\end{array}
\]
\end{lemma}
The other lemma is more involved: 
\begin{lemma}[\cite{romash-cond-indep}, see also  exercise~316 in \cite{shen-vereshchagin}]\label{l:cond-independence}
There exists an integer number  $k$ with the following property.
Let $x=x_0$ and $y=y_0$ be a pair of strings from Example~\ref{hamming} (two $n$-bit strings with the Hamming distance $\theta n$ and complexity profile \eqref{eq:complexity-profile}). Then there exist  two sequences of $n$-bit binary strings $x_1,\ldots,x_k$ and $y_1,\ldots,y_k$ such that 
\begin{itemize}
\item $I(x_i:y_i|x_{i+1}) = O(\log n)$ for $i=0,\ldots,k-1$,
\item $I(x_i:y_i|y_{i+1}) = O(\log n)$ for $i=0,\ldots,k-1$,
\item $I(x_k:y_k) =O(\log n)$.
\end{itemize}
(Note that $k$ is a constant that does not depend on $n$. It is  determined uniquely by the value of $\theta$, which is the normalized Hamming distance between $x$ and $y$.)
\end{lemma}
\begin{remark}
In the proof Lemma~\ref{l:cond-independence} suggested in \cite{romash-cond-indep}, each pair $(x_i,y_i)$ consists of two binary strings of length $n$ with Hamming distance $\theta_i n$ and maximal possible (for this value of $\theta_i$) Kolmogorov complexity. In our case, the initial $\theta_0=\theta$ is chosen so that $I(x_0:y_0) \eqp n/2$. Each next $\theta_i$ is bigger than the previous one. For the last pair  we have $\theta_k=1/2$. This means that  in the last pair $(x_k,y_k)$   the strings differ in a half of the positions, so the mutual information is  only $O(\log n)$.
\end{remark}
Applying  Lemma~\ref{l:ineq}(ii), we obtain for every string $w$ and for all $x_i,y_i, x_{i+1}, y_{i+1}$ the inequalities
\[
\begin{array}{ccl}
C(w|x_{i+1})& \lep &C(w|x_{i}) + C(w|y_{i}) + I(x_i:y_i | x_{i+1} ),\\
C(w|y_{i+1})& \lep &C(w|x_{i}) + C(w|y_{i}) + I(x_i:y_i | y_{i+1} ).\\
\end{array}
\]
Combining these inequalities for $i=0,\ldots,k-1$ and taking into account  the assumptions $I(x_i:y_i|x_{i+1}) = O(\log n)$  and  $I(x_i:y_i|y_{i+1}) = O(\log n)$,  we obtain
\[
C(w|x_{k}) + C(w|y_{k}) \lep 2^k \cdot \left( C(w|x_{0}) + C(w|x_{0}) \right).
\]
Now we use Lemma~\ref{l:ineq}(i)  and obtain
\[
C(w) \lep C(w|x_{k}) + C(w|y_{k}) + I(x_k:y_k).
\]
With the condition $I(x_k:y_k) =O(\log n)$ we get
\begin{equation}\label{eq:2^k}
C(w)\lep 2^k\cdot \left(C(w|x_0)+C(w|y_0) \right) +I(x_k:y_k)\eqp 2^k\cdot \left(C(w|x)+C(w|y)\right) .
\end{equation}
Denote by $r_A$ and $r_B$ the strings of random bits used in the protocol by Alice and Bob respectively. With high probability the randomly chosen $r_A$ and $r_B$ have negligibly small mutual information with $w,x,y$. Therefore, \eqref{eq:2^k} rewrites to
\begin{equation}\label{eq:2^k-bis}
C(w)\lep 2^k\cdot\left(C(w|x,r_A)+C(w|y,r_B)\right) .
\end{equation}

We apply \eqref{eq:2^k-bis} to $w:=\langle t,z\rangle$, where $z$ is the secret key obtained by Alice and Bob, and
$t$ is the transcript of the communication protocol. Then 
\[
C(w) \eqp C(z) + C(t)
\]
(the key has no mutual information with the transcript), and
\[
C(w|x,r_A) \lep C(t),\ C(w|y,r_B) \lep C(t)
\]
(given the transcript and the data available to Alice or to Bob, we can compute $z$). Plugging this in \eqref{eq:2^k-bis} we obtain
\[
C(z) + C(t) \lep 2^{k+1} \cdot  C(t) ,
\]
which implies $C(t) = \Omega(C(z))$. Therefore, the size of the transcript $t$  is not less than $ \Omega(\delta n)$, and we are done.
\end{proof}

\section{Conclusion}

In Theorem~\ref{th:main} we have proven a lower bound for communication complexity of protocols with \emph{private} randomness. The argument can be extended to the setting where Alice and Bob use both \emph{private} and \emph{public} random bits (the private sources of randomness are available only to Alice and Bob respectively; the public source of randomness is available to both parties and to the eavesdropper). Thus, the problem of the \emph{worst case}  complexity is resolved for the most general natural model of communication.

For pairs of inputs $(x,y)$ from Example~\ref{lines-and-points} and Example~\ref{quasi-distance} 
(corresponding to edges of a graph with a large spectral gap) we have  got a lower bound for communication complexity that matches the known upper bound.
However,  in the known protocols achieving this bound the communication is very non-symmetric: all (or almost all) the burden of communication falls on only one of two participants. We do not know whether 
there exist more balanced protocols, where the communication complexity is shared equally between Alice and Bob.

We have no characterization of  the optimal communication complexity of the  secret key agreement for pairs of inputs $(x,y)$ that do not enjoy the spectral property required in Corollary~\ref{mixing-lemma-large-sets}. In particular, there remains a large gap between  constant hidden in the $O(\delta n)$ notation in Theorem~\ref{th:hamming1}  and in the  $\Omega(\delta n)$ notation in  and Theorem~\ref{th:hamming2}, so the question on the optimal trade-off between the secret key size and communication complexity for $(x,y)$ from Example~\ref{hamming} remains open (\emph{cf}.  Conjecture~1 in \cite{Liu-Cuff-Verdu} for an analogous problem in Shannon's setting).

In some applications (see, e.g.,  \cite{devetak2005distillation,horodecki2009quantum}) it is natural to assume that the eavesdropper is given a non-negligible \emph{a priori} information about Alice' and Bob's inputs. For this setting we do not have a characterization of the optimal size of the key and the communication complexity of the protocol.

\paragraph{Acknowledgements}
The authors thank the anonymous reviewers 
for instructive suggestions and corrections concerning this conference publication as well as the full version of the paper.
The second author is grateful to  the Max Planck Institute for Mathematics in the Sciences (Leipzig, Germany) for hospitality,
and  thanks Rostislav Matveev and Jacobus Portegies  for useful discussions.

\bibliographystyle{plain}
\bibliography{arxiv-v6}{}

\appendix

\section{Information-theoretic security in the language of Kolmogorov complexity.}\label{s:antunes-theory}

In this section we briefly recall the notion of secure encoding in the classical settings (in terms of probability distributions and Shannon's information theory)
and its homologues in the framework of Kolmogorov complexity, as proposed in  \cite{antunes2007cryptographic}. 
We begin with the very basic cryptographic primitive---a symmetric encryption scheme.
\begin{definition}
We say that a private-key encryption scheme is a pair of algorithms $(\mathrm{Enc}_n, \mathrm{Dec}_n)$ and associated spaces: 
the set of clear messages ${\mathcal M}_n$, the set of ciphertexts ${\mathcal C}_n$, and for the space of keys ${\mathcal K} = \{0,1\}^n$,
\[
\begin{array}{l}
\mathrm{Enc}_n : {\mathcal M}_n\times{\mathcal K}_n \to {\mathcal C}_n\\
\mathrm{Dec}_n : {\mathcal C}_n \times {\mathcal K}_n \to {\mathcal M}_n
\end{array}
\]
such that for every $m\in \mathcal M$ and every $k\in {\mathcal K}_n$ we have
$
\mathrm{Dec}_n(\mathrm{Enc}_n(m,k),k) = m
$
(encrypting a message and then decrypting the resulting ciphertext with the same key yields the original message).
\end{definition}

There is a pretty common definition of a perfectly secure encryption scheme:

\begin{definition}\label{def:perfect-secrecy}
A private-key encryption scheme $(\mathrm{Enc}_n, \mathrm{Dec}_n)$  is called  \emph{perfectly secure} if for the uniform distribution of messages on ${\mathcal M}_n$ 
and the uniform distribution of the key on ${\mathcal K} $ (independent of the distribution of the messages $m$), 
we have
\[
\forall m\in {\mathcal M}_n,\ \forall e\in {\mathcal C}_n, \   \mathrm{Pr}[\text{the clear message is } m \mid  {\mathrm Enc}_n(m,k) = e]   =  \mathrm{Pr} [\text{the clear message is } m] 
\]
In words: the knowledge of the ciphertext does not bring any information on the message.
\end{definition}

\begin{remark} 
We can easily reformulate the definition for more general (non-uniform) distributions on ${\mathcal M}_n$.
\end{remark}

\begin{example}
The classical Vernam cipher scheme (one-time pad) where $\mathrm{Enc}(m,k) = m\oplus k$ (bitwise XOR) is \emph{perfectly secure}.
\end{example}

The assumption of independence of two random variables is a clean and nice mathematical concept, it is a perfect starting point for building an elegant mathematical theory.
But, in practice, we hardly can guarantee that two random values are absolutely independent.
Let us define a relaxed version of Definition~\ref{def:perfect-secrecy}, with a minor imprecision allowed  in the condition of independence:

\begin{definition}[see \cite{antunes2007cryptographic}]\label{def:almost-perfect-secrecy}
A private-key encryption scheme with a distribution $\mu$ defined  on ${\mathcal M}_n \times {\mathcal K}_n$  
is called $\sigma$-\emph{almost perfectly secure} if 
\[
H(\text{random value of }m)  - H(\text{random value of }m \mid \text{random value of } e) <\sigma,
\]
where $e = {\mathrm Enc}_n(m,k)$, and 
$H(\cdot)$ stands for Shannon entropy.
In what follows we assume for simplicity that $\mu$ is uniform, although this definition can be used in much more general settings.
\end{definition}
This definition of security has a homologue in the framework of Kolmogorov complexity,
where we define not security of the scheme in general but security of a specific ciphertext for a specific message.

\begin{definition}[again, see \cite{antunes2007cryptographic}]\label{def:k-almost-perfect-secrecy}
Let us fix  a private-key encryption scheme and a distribution $\mu$ on ${\mathcal M}_n \times {\mathcal K}_n$. We say that for an instance $(m,k)\in{\mathcal M}_n\times {\mathcal K}_n$ the ciphertext $e=\mathrm{Enc}(m,k)$ is $\rho$-secure if 
\[
C(m) - C(m \mid e) < \rho.
\]
In words: the encryption is secure  if the ciphertext contains negligibly small information on the message. 
\end{definition}

\begin{example} 
In the  Vernam cipher scheme, XOR-ing the clear message $m$ with a Kolmogorov-random (independent of $m$) key $k$ results in a secure ciphertext; 
XOR-ing the clear message $m$ with an instance of key $k=000\ldots0$ is not secure, see \cite[Section~3.3]{antunes2007cryptographic} for more details.
\end{example}

Definition~\ref{def:almost-perfect-secrecy} and Definition~\ref{def:k-almost-perfect-secrecy} are related with each other  in some  formal sense.
If a scheme is almost perfectly secure (Definition~\ref{def:almost-perfect-secrecy}), 
then for most instances $(m,k)\in{\mathcal M}_n\times {\mathcal K}_n$ the ciphertext  is secure (Definition~\ref{def:k-almost-perfect-secrecy}).
There is also an implication in the opposite direction:
if many  instances of a cipher system are secure (Definition~\ref{def:k-almost-perfect-secrecy}), then the system is almost perfectly secure (Definition~\ref{def:almost-perfect-secrecy}). 
Moreover, the connection between the definitions of security in Shannon's and Kolmogorov's frameworks can be extended to non-uniform distributions.
The precise statement of the translation between Definition~\ref{def:almost-perfect-secrecy} and Definition~\ref{def:k-almost-perfect-secrecy}
(the trade-off between $\sigma$ in Definition~\ref{def:almost-perfect-secrecy} and $\rho$ in Definition~\ref{def:k-almost-perfect-secrecy}) 
can be found in \cite[Theorem~13]{antunes2007cryptographic}.

A number of other cryptographic definitions can also be translated into the language of Kolmogorov complexity, see, e.g., the discussion of secret sharing and authentication codes in 
\cite{antunes2007cryptographic}.

\section{Technical Lemmas and a Proof of Proposition~\ref{p:newman}}

In this section we prove  a version of Newman's theorem (on sampling random bits) 
for communication protocols with private randomness.
We start with two technical lemmas and then proceed with a proof of  Proposition~\ref{p:newman}.

\begin{definition}
Let $R_A$ and $R_B$ be finite sets and let $S\subset A\times B$. We say that a sequence (or a multiset) of elements $a_1,\ldots,a_k$ in $A$ and a sequence (multiset) of elements $b_1,\ldots,b_k$ in $B$ provide a $\delta$-precise sampling of $S$ in $A\times B$ if
\[
 \left|\prob_{i,j}[(a_i,b_j)\in S]-\frac{|S|}{|A\times B|}\right|<\delta,
\]
see Fig.~\ref{pic:duble-sampling}.
\end{definition}

\begin{figure}
\begin{center}
\begin{tikzpicture}[scale=0.70, axis/.style={very thick, ->}, thick/.style={line width=1.5pt}]

\draw [black] plot [smooth cycle] coordinates {(7,-1) (9.2,-1.4) (8,1.5) (4,5.1) (3.5,4.5)} ;

\fill[color=gray!70]{ plot [smooth cycle] coordinates {(7,-1) (9.2,-1.4) (8,1.5) (4,5.1) (3.5,4.5)} };

\draw[axis] (1.0,-3)  -- (11,-3) node(xline)[right] {};
\draw[axis] (1.5,-3.5)  -- (1.5,6.0) node(xline)[right] {};
   
\draw [decorate,decoration={brace,amplitude=10pt},thick]  (9.21,-3.8)--(3.26,-3.8) ;
 \draw  (6.20,-4.70) node { $A$ };
   \draw  (3.26,-3.5) -- (3.26,5.8);    
   \draw  (9.2,-3.5) -- (9.25,5.8);   

   \draw [decorate,decoration={brace,amplitude=10pt},thick]  (0.5,-1.65) -- (0.5,5.25) ;
   \draw (0.6,5.35) -- (9.8,5.35);       
   \draw (0.6,-1.65) -- (9.8,-1.65);   
 \draw  (-0.5,1.75) node { $B$ };

 \draw [dashed] (4.0,-3.0) -- (4.0,5.35);   
 \draw [thick]  node at  (4.0,-3.0) {$\bullet$}; 
 \draw  (4.0,-3.5) node { $a_1$ };

 \draw [dashed] (5.0,-3.0) -- (5.0,5.35);   
 \draw [thick]  node at  (5.0,-3.0) {$\bullet$}; 
  \draw  (5.0,-3.5) node { $a_2$ };

 \draw [dashed] (5.5,-3.0) -- (5.5,5.35);   
 \draw [thick]  node at  (5.5,-3.0) {$\bullet$}; 
 \draw  (5.5,-3.5) node { $a_3$ };

 \draw [dashed] (6.25,-3.0) -- (6.25,5.35);   
 \draw [thick]  node at  (6.25,-3.0) {$\bullet$};  
  \draw  (6.25,-3.5) node { $a_4$ };

 \draw  (7.3,-3.5) node { $\ldots$ };

 \draw [dashed] (8.3,-3.0) -- (8.3,5.35);   
 \draw [thick]  node at  (8.3,-3.0) {$\bullet$};  
 \draw  (8.3,-3.5) node { $a_k$ };

 \draw [dashed] (1.5,-1.1) -- (9.3,-1.1);   
 \draw [thick]  node at  (1.5,-1.1) {$\bullet$}; 
  \draw  (1.1,-1.1) node { $b_1$ };

 \draw [dashed] (1.5,-0.3) -- (9.3,-0.3);   
 \draw [thick]  node at  (1.5,-0.3) {$\bullet$}; 
   \draw  (1.1,-0.3) node { $b_2$ };

 \draw [dashed] (1.5,0.3) -- (9.3,0.3);   
 \draw [thick]  node at  (1.5,0.3) {$\bullet$}; 
  \draw  (1.1,0.3) node { $b_3$ };

 \draw [dashed] (1.5,1.7) -- (9.3,1.7);   
 \draw [thick]  node at  (1.5,1.7) {$\bullet$}; 
  \draw  (1.1,1.7) node { $b_4$ };

 \draw [dashed] (1.5,2.7) -- (9.3,2.7);   
 \draw [thick]  node at  (1.5,2.7) {$\bullet$}; 
  \draw  (1.1,2.6) node { $b_5$ };

 \draw [dashed] (1.5,3.0) -- (9.3,3.0);   
 \draw [thick]  node at  (1.5,3.0) {$\bullet$}; 
  \draw  (1.1,3.1) node { $b_6$ };
  
  \draw  (1.1,4.0) node { $\vdots$ };
  
 \draw [dashed] (1.5,4.6) -- (9.3,4.6);   
 \draw [thick]  node at  (1.5,4.6) {$\bullet$}; 
   \draw  (1.1,4.6) node { $b_k$ };

 \draw [thick]  node at  (4.0,4.6) {$\bullet$}; 
 \draw [thick]  node at  (5.0,3.0) {$\bullet$}; 
 \draw [thick]  node at  (5.0,2.7) {$\bullet$}; 
 \draw [thick]  node at  (5.5,3.0) {$\bullet$}; 
 \draw [thick]  node at  (5.5,2.7) {$\bullet$}; 
 \draw [thick]  node at  (6.25,3.0) {$\bullet$}; 
 \draw [thick]  node at  (6.25,2.7) {$\bullet$}; 
 \draw [thick]  node at  (5.5,1.7) {$\bullet$};  
 \draw [thick]  node at  (6.25,1.7) {$\bullet$};  
 \draw [thick]  node at  (6.25,0.3) {$\bullet$};  
 \draw [thick]  node at  (8.3,0.3) {$\bullet$};  
 \draw [thick]  node at  (8.3,-0.3) {$\bullet$}; 
 \draw [thick]  node at  (8.3,-1.1) {$\bullet$};

  \draw [thick]  node at  (5.0,4.6) {$\circ$}; 
  \draw [thick]  node at  (5.5,4.6) {$\circ$};  
  \draw [thick]  node at  (6.25,4.6) {$\circ$};     
  \draw [thick]  node at  (8.3,4.6) {$\circ$};   

  \draw [thick]  node at  (4.0,3.0) {$\circ$}; 
  \draw [thick]  node at  (8.3,3.0) {$\circ$};   

  \draw [thick]  node at  (4.0,2.7) {$\circ$}; 
  \draw [thick]  node at  (8.3,2.7) {$\circ$};   

  \draw [thick]  node at  (4.0,1.7) {$\circ$}; 
  \draw [thick]  node at  (5.0,1.7) {$\circ$}; 
  \draw [thick]  node at  (8.3,1.7) {$\circ$};

  \draw [thick]  node at  (4.0,0.3) {$\circ$}; 
  \draw [thick]  node at  (5.0,0.3) {$\circ$}; 
  \draw [thick]  node at  (5.5,0.3) {$\circ$};

  \draw [thick]  node at  (4.0,-0.3) {$\circ$}; 
  \draw [thick]  node at  (5.0,-0.3) {$\circ$}; 
  \draw [thick]  node at  (5.5,-0.3) {$\circ$};  
  \draw [thick]  node at  (6.25,-0.3) {$\circ$};

  \draw [thick]  node at  (4.0,-1.1) {$\circ$}; 
  \draw [thick]  node at  (5.0,-1.1) {$\circ$}; 
  \draw [thick]  node at  (5.5,-1.1) {$\circ$};  
  \draw [thick]  node at  (6.25,-1.1) {$\circ$};

 \draw  (7.2,1.0) node { \Huge $\mathbf{S}$ };
\end{tikzpicture}
\caption{Sampling random points in $S\subset A\times B$ by choosing independently random columns $a_i\in A$ and random rows $b_j\in B$.}\label{pic:duble-sampling}
\end{center}
\end{figure}
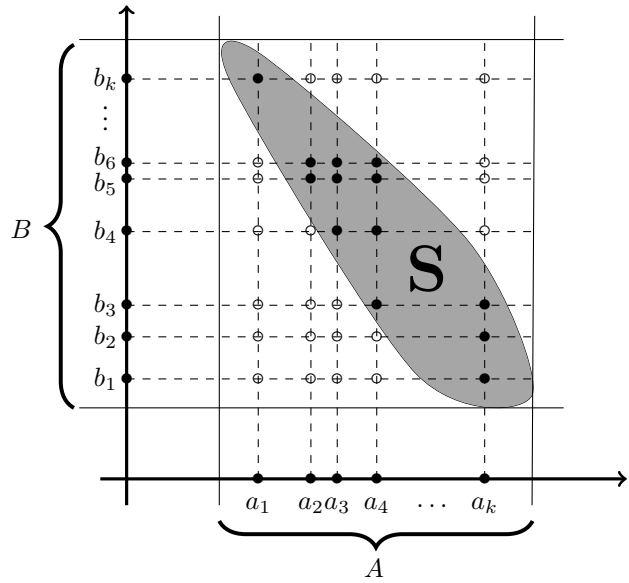

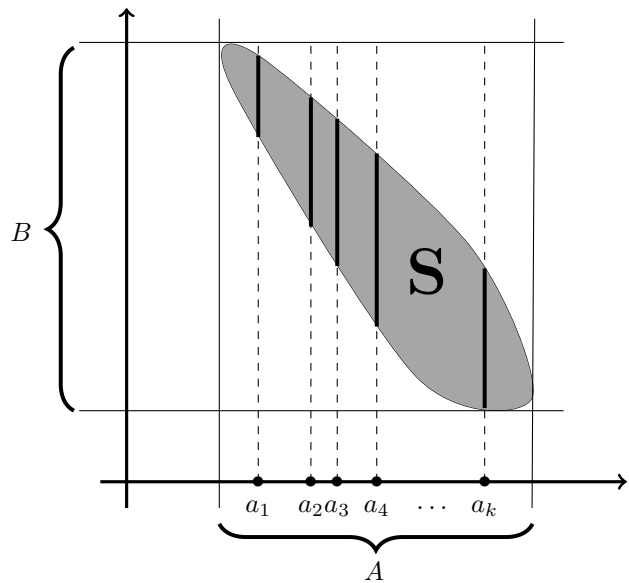
\begin{figure}
\begin{center}
\begin{tikzpicture}[scale=0.70, axis/.style={very thick, ->}, thick/.style={line width=1.5pt}]

\draw [black] plot [smooth cycle] coordinates {(7,-1) (9.2,-1.4) (8,1.5) (4,5.1) (3.5,4.5)} ;

\fill[color=gray!70]{ plot [smooth cycle] coordinates {(7,-1) (9.2,-1.4) (8,1.5) (4,5.1) (3.5,4.5)} };

\draw[axis] (1.0,-3)  -- (11,-3) node(xline)[right] {};
\draw[axis] (1.5,-3.5)  -- (1.5,6.0) node(xline)[right] {};

\draw [decorate,decoration={brace,amplitude=10pt},thick]  (9.21,-3.8)--(3.26,-3.8) ;
 \draw  (6.20,-4.70) node { $A$ };
   \draw  (3.26,-3.5) -- (3.26,5.8);    
   \draw  (9.2,-3.5) -- (9.25,5.8);   

   \draw [decorate,decoration={brace,amplitude=10pt},thick]  (0.5,-1.65) -- (0.5,5.25) ;
   \draw (0.6,5.35) -- (9.8,5.35);       
   \draw (0.6,-1.65) -- (9.8,-1.65);   
 \draw  (-0.5,1.75) node { $B$ };

 \draw [dashed] (4.0,-3.0) -- (4.0,5.35); 
 \draw [thick]  node at  (4.0,-3.0) {$\bullet$}; 
 \draw  (4.0,-3.5) node { $a_1$ };
  \draw [thick] (4.0,3.55) -- (4.0,5.1);

 \draw [dashed] (5.0,-3.0) -- (5.0,5.35);   
 \draw [thick]  node at  (5.0,-3.0) {$\bullet$}; 
  \draw  (5.0,-3.5) node { $a_2$ };
  \draw [thick] (5.0,1.85) -- (5.0,4.30);   

 \draw [dashed] (5.5,-3.0) -- (5.5,5.35);   
 \draw [thick]  node at  (5.5,-3.0) {$\bullet$}; 
 \draw  (5.5,-3.5) node { $a_3$ };
  \draw [thick] (5.5,1.1) -- (5.5,3.90); 

 \draw [dashed] (6.25,-3.0) -- (6.25,5.35);   
 \draw [thick]  node at  (6.25,-3.0) {$\bullet$};  
  \draw  (6.25,-3.5) node { $a_4$ };
    \draw [thick] (6.25,-0.05) -- (6.25,3.24); 

 \draw  (7.3,-3.5) node { $\ldots$ };

 \draw [dashed] (8.3,-3.0) -- (8.3,5.35);   
 \draw [thick]  node at  (8.3,-3.0) {$\bullet$};  
 \draw  (8.3,-3.5) node { $a_k$ };
    \draw [thick] (8.3,-1.6) -- (8.3,1.05);

 \draw  (7.2,1.0) node { \Huge $\mathbf{S}$ };
\end{tikzpicture}
\caption{Restriction of a set $S\subset A\times B$ onto the product $\{a_1,\ldots,a_k\}\times B$
with randomly chosen  columns $a_i\in A$.}
\label{pic:sinlge-sampling}
\end{center}
\end{figure}

\begin{lemma}\label{l:large-numbers}
Let $A$ and $B$ be finite sets and let $S\subset A\times B$. 
We choose at  random sequences of elements $a_1,\ldots,a_k$ in $A$ and  $b_1,\ldots,b_k$ in $B$
(all $a_i$ and $b_j$ are independent and uniformly distributed on $A$ and $B$ respectively).
Then with probability at least 
$
1-4e^{-\delta^2 k/2}
$
the chosen pair of sequences provides  a $\delta$-precise sampling  of $S$.
\end{lemma}
\begin{proof}
To prove the lemma, we split the process of sampling in two steps. At first, we substitute $A$ by a sequence $a_1,\ldots,a_k$ and estimate the difference
 \begin{equation}
 \label{eq:derandom1}
 \left|\prob_{i\in\{1,\ldots,k\},b\in B}[(a_i,b)\in S]-\prob_{a\in A,b\in B}[(a,b)\in S]\right|,
\end{equation}
see Fig.~\ref{pic:sinlge-sampling}.
Then, we  fix a sequence $a_1,\ldots,a_k$, choose at random a sequence $b_1,\ldots,b_k\in B$ and estimate the difference 
\begin{equation}
 \label{eq:derandom2}
 \left|\prob_{i\in\{1,\ldots,k\},j\in\{1,\ldots,k\}}[(a_i,b_j)\in S]-\prob_{i\in\{1,\ldots,k\},b\in B}[(a_i,b)\in S]\right|.
\end{equation}
We show that both of these differences are typically small and, therefore, the difference 
\[
 \left|\prob_{i\in\{1,\ldots,k\},j\in\{1,\ldots,k\}}[(a_i,b_j)]-\prob_{a\in A,b\in B}[(a,b)\in S]\right|  
\]
is typically small as well.

\smallskip
\noindent
\emph{Claim:}
For a randomly chosen sequence $a_1,\ldots,a_k\in A$, the probability that the difference in \eqref{eq:derandom1} is larger than $\delta/2$ is less than $2e^{-\delta^2 k/2}$.
\smallskip

\begin{proof}[Proof of the claim:]
For every $a\in A$ we denote by $\varphi(a)$ the fraction of elements from $S$ in the column of the Cartesian product $A\times B$ 
corresponding to the value $a$, i.e., 
 \[
 \varphi(a) := \frac{|S\cap (\{a\}\times B )|}{|B|}.
 \]
 From the definition it follows that $0\le \varphi(a) \le 1$ for every $a\in A$, and the average value of $\varphi(a)$ (for a randomly chosen $a$) is equal to $\sigma:=\frac{|S|}{|A \times B|}$.
Since $a_1,\ldots,a_k$ are chosen in  $A$ independently, the values $\varphi(a_1),\ldots,\varphi(a_k)$ are independent and identically distributed. Hence, we can apply Hoeffding's  inequality (see \cite{hoeffding1963probability}):
\[
  \prob\left[\left|\frac{\varphi(a_1)+\ldots+\varphi(a_k)}{k}-\sigma\right|\ge\delta/2\right]\le2e^{-2(\delta/2)^2k}.
\]
This means that with a probability $\ge 1-2e^{-2(\delta/2)^2 k}$  the gap \eqref{eq:derandom1} is less than $\delta/2$.
 \end{proof}
 
Let us fix now a sequence $a_1,\ldots,a_k\in A$, and estimate the difference in \eqref{eq:derandom2} as a function of random sequence $b_1,\ldots,b_k\in B$.
We can use once again the same argument with Hoeffding's inequality as in Claim above, now for the sampling of values of the second coordinate in $\{a_1,\ldots,a_k\}\times B$.  
It follows that with a probability $\ge1-2e^{-\delta^2k/2}$ the gap \eqref{eq:derandom2} is less than $\delta/2$.
Combining the bounds for \eqref{eq:derandom1} and\eqref{eq:derandom2}, we obtain the lemma.
\end{proof}

\begin{lemma}\label{l:derandom}
Let $\pi$ be a two-party communication protocol where Alice and Bob access inputs $x$ and $y$ respectively, 
and use  private random bits $r_A\in \{0,1\}^s$ and  $r_B\in \{0,1\}^s$ respectively. For $z_1,z_2,t$ we denote
by $ p(x,y,z_1,z_2,t) $ the probability of the following event: given $x$ and $y$ as inputs, Alice and Bob apply the protocol $\pi$
and end up with an answer $z_1\in \{0,1\}^m$ for Alice, an answer $z_2\in \{0,1\}^m$ for Bob, and a communication transcript $t\in \{0,1\}^l$. The probability is taken on independent and uniformly distributed $r_A, r_B$.

 We apply to $\pi$ a random transformation as follows: we choose two random sequences 
   $a_1,\ldots,a_k$ and $b_1,\ldots,b_k$ in $\{0,1\}^s$.
 In the new communication protocol $\pi'$ Alice and Bob choose at random $i,j\in\{1,\ldots,k\}$, and then apply the original communication protocol with private strings of random bits $a_i$ and $b_j$ respectively. By construction, in the protocol $\pi'$ Alice and Bob need  $\log k$ private random bits each (they both need to choose a random index between $1$ and $k$).

Denote by $ p'(x,y,z_1,z_2,t) $ the probability analogous to $ p(x,y,z_1,z_2,t) $ computed for the protocol $\pi'$. The probability is taken on the choice of 
independent and  uniformly distributed $i,j\in \{1,\ldots,k\}$.
We claim that for all $x,y,z_1,z_2,t$ the probability of the event 
\begin{equation}\label{eq:derandom3}
|p(x,y,z_1,z_2,t)  - p'(x,y,z_1,z_2,t)  | < \delta 
\end{equation}
is greater than $1-4e^{-\delta^2 k/2}$ (here the probability is taken over the choice of $a_1,\ldots,a_k$ and $ b_1,\ldots,b_k$ in $\{0,1\}^s$
in the construction of $\pi'$).
\end{lemma}
\begin{proof}
We apply Lemma~\ref{l:large-numbers} with $A=B=\{0,1\}^s$ (the space of random bits of the original protocol), and $S=S(x,y,z_1,z_2,t)$ that consists of the pairs $(r_A,r_B)\in\{0,1\}^s\times \{0,1\}^s$ compatible with the given $x,y,z_1,z_2,t$. In other words, a pair $(r_A,r_B)$ belongs to $S$, if Alice and Bob given $x$ and $y$ as inputs and $r_A$ and $r_B$ as random bits, obtain with the protocol $\pi$ answers $z_1$ and $z_2$ respectively and produce a communication transcript $t$. From the lemma it follows that with the probability $1-4e^{-\delta^2 k/2}$, the choice of $a_1,\ldots,a_k$ and $b_1,\ldots,b_k$ in $\{0,1\}^s$ in the construction of $\pi'$ provides a $\delta$-precise sampling of $S$.
\end{proof}

Now we are ready to prove a version of Newman's theorem suitable for our setting.

\begin{proof}[Proof of Proposition~\ref{p:newman}] 
Due to  Lemma~\ref{l:derandom}, with  randomly chosen  samples $a_1,\ldots,a_k$ and $b_1,\ldots,b_k$ 
 with a probability $>1-4e^{-\delta^2k/2}$ we obtain a new protocol 
 $\pi'$ such that  for each $(x,y,z_1,z_2,t)$ we have  \eqref{eq:derandom3}. 
We apply this construction with a $\delta$ such that 
 
\[
\delta : = \frac{\varepsilon_2 }{ \# \text{answers }z_1 \cdot  \# \text{answers }z_2  \cdot \#  \text{transcripts }t}
\]
and a $k$ such that
\begin{equation}\label{eq:delta1}
4e^{-\delta^2 k/2} \cdot  \# \text{inputs }(x,y)  \cdot  \# \text{answers }z_1 \cdot  \# \text{answers }z_2  \cdot \#  \text{transcripts }t <1.
\end{equation}
The property \eqref{eq:delta1} implies that there are samples  $a_1,\ldots,a_k$ and $b_1,\ldots,b_k$
such that  \eqref{eq:derandom3} is true simultaneously for all $(x,y,z_1,z_2,t)$. Let us fix one such instance of sampling.

In the obtained protocol $\pi'$, for each pair of input $(x,y)$ the total probability to obtain an  invalid outcome (which is the sum over all $(z_1,z_2,t)$ that are invalid) increases in comparison with the original protocol $\pi$ by at most
\begin{equation}\label{eq:delta2}
\delta \cdot { \# \text{answers }z_1 \cdot  \# \text{answers }z_2  \cdot \#  \text{transcripts }t} < \varepsilon_2.
\end{equation}
Note that we can claim \eqref{eq:delta2} even without having an explicit description of  the sets of valid and invalid outcomes,
it is enough to have  \eqref{eq:derandom3} for each tuple of inputs and outcomes.

Since the length of the outputs and the communication complexity of $\pi$ are linear in $n$, the total number of inputs $(x,y)$, outputs $(z_1,z_2)$, and transcripts $t$ is $2^{O(n)}$. It is not hard to see that  we can chose
\[
\delta = \varepsilon_2/ 2^{O(n)}
\text{ and }
k = 2^{O(n)}/\varepsilon_2^2
\]
so that \eqref{eq:delta1} and \eqref{eq:delta2} are satisfied.
From the choice of $k$ it follows that in  the new protocol $\pi'$ Alice and Bob need only  $\log k = O(n + \log(1/\varepsilon_2))$ private random bits.

It remains to notice that an appropriate instance of sampling in Lemma~\ref{l:derandom} (suitable sequences $a_1,\ldots,a_k$ and $b_1,\ldots,b_k$)  can be found by brute-force search. 
To organize this search we do not need to know the precise definition of the validity of the outcome nor the class of admissible inputs $(x,y)$. Indeed, we can simulate the original protocol $\pi$ on \emph{all} pairs of inputs of length $n$ and find an instance of sampling such that for each $(z_1,z_2,t)$ the probability of this outcome increases by at most $\delta$.
Thus, if the original protocol was uniformly computable, so is the new one (though the computational complexity could increase dramatically).
\end{proof}

\end{document}